\documentclass[11pt]{article}

\usepackage{amsmath,amssymb,amsthm,comment,graphicx,lineno,color}
\usepackage{algorithm}
\usepackage[numbers]{natbib}
\usepackage{adjustbox}
\usepackage{caption}    
\usepackage{setspace}
\usepackage{subcaption}
\usepackage{placeins}
\usepackage{float}
\usepackage[utf8]{inputenc}
\usepackage{multirow}
\usepackage{mathrsfs}

\theoremstyle{definition}
\newtheorem{defn}{Definition}

\theoremstyle{plain}
\newtheorem{theorem}{Theorem}

\newtheorem{lemma}[theorem]{Lemma}

\providecommand{\keywords}[1]{\textbf{\textrm{Keywords:}} #1}                          
\setcitestyle{square} 
 \usepackage[affil-it]{authblk}

\newcommand{\FDP}{\textrm{FDP}}
\newcommand{\FDR}{\textrm{FDR}}
\newcommand{\FNR}{\textrm{FNR}}
\newcommand{\Lfdr}{\textrm{Lfdr}}
\newcommand{\mFDR}{\textrm{mFDR}}
\newcommand{\mFNR}{\textrm{mFNR}}
\newcommand{\EE}{\mathbb{E}}
\newcommand{\PP}{\mathbb{P}}
\newcommand{\bigO}{\mathcal{O}}

\newcommand\smallO{
  \mathchoice
    {{\scriptstyle\mathcal{O}}}
    {{\scriptstyle\mathcal{O}}}
    {{\scriptscriptstyle\mathcal{O}}}
    {\scalebox{.7}{$\scriptscriptstyle\mathcal{O}$}}
  }

\title{A Generalized Benjamini-Hochberg Procedure for Multivariate Hypothesis Testing}

\author[*]{Kasra Alishahi}
\author[*]{Ahmad Reza Ehyaei}
\author[**]{Ali Shojaie}
\affil[*]{Department of Mathematical Sciences, Sharif University of Technology, Tehran, Iran}
\affil[**]{Department of Biostatistics, University of Washington, Seattle, WA}
\date{\today}

\begin{document}
\maketitle

\begin{abstract}
The introduction of the false discovery rate (FDR) by Benjamini and Hochberg has spurred a great interest in developing methodologies to control the FDR in various settings. The majority of existing approaches, however, address the FDR control for the case where an appropriate univariate test statistic is available. Modern hypothesis testing and data integration applications, on the other hand, routinely involve multivariate test statistics. The goal, in such settings, is to combine the evidence for each hypothesis and achieve greater power, while controlling the number of false discoveries. This paper considers data-adaptive methods for constructing nested rejection regions based on multivariate test statistics ($z$-values). It is proved that the FDR can be controlled for appropriately constructed rejection regions, even when the regions depend on data and are hence random. This flexibility is then exploited to develop optimal multiple comparison procedures in higher dimensions, where the distribution of non-null $z$-values is unknown. Results are illustrated using simulated and real data. 
\end{abstract}


\section{Introduction}\label{sec:intro}

Multiple hypothesis testing is a fundamental problem in many new scientific applications involving Big Data, from genomics and neuroscience to astronomy and finance. 
The false discovery rate (FDR) control, introduced in Benjamini and Hochberg's seminal paper \citep{BH95}, is one of the most important methodological developments in multiple hypothesis testing. 
To control the FDR at a predetermined level $\alpha$, \citeauthor{BH95} (BH) proposed a step-down procedure based on ranked $p$-values. The procedure was initially developed under the assumption that the $n_0$ null hypotheses are independent of each other and are also independent of the $n_1$ non-null hypotheses; nonetheless, it continues to control the FDR if the $n=n_0+n_1$ hypotheses are positively dependent; see, e.g., \citet{BY01} for additional details. 

A potential drawback of the BH procedure is that it controls the FDR at the level $(n_{0}/n)\alpha$ and is hence conservative. 
A number of authors, including, \citet{BH00}, \citet{STS04}, \citet{BKY06} and \citet{GBS09}, have thus proposed modifications of the BH procedure for more efficient multiple testing by estimating $n_0$. 
These `BH-type' procedures primarily focus on better control of the FDR, but their efficiency has not been formally investigated. 
In particular, while they guarantee the control of $\FDR=\EE(V/R)$, for $V$ and $R$ defined in Table~\ref{tab:classification}, they do not provide any guarantees on the false negative rate $\FNR=\EE\left(T/(n-R)\right)$. 
\begin{table}[b]
\centering
\caption{Possible outcomes in multiple hypothesis testing problems.}\label{tab:classification}
\begin{adjustbox}{max width=\textwidth}
    \begin{tabular}{| l | c | c | c |}
\hline 
\, & \multicolumn{2}{c|}{Decisions} & \, \\ 
\cline{2-3} 
Truth & Not rejected & Rejected & Total \\ 
\hline 
Null & $U$ & $V$ & $n_0$ \\ 
\hline 
Non-Null & $T$ & $S$ & $n_1$ \\ 
\hline 
Total & $n-R$ & $R$ & $n$ \\ 
\hline 
    \end{tabular}
      \centering
 \end{adjustbox}
\end{table}

As an alternative to controlling the FDR based on $p$-values, a number of authors have advocated the use of test statistics, or $z$-values. 
\citet{Eb07} introduced the control of \emph{local} false discovery rate ($\Lfdr$), which facilitates the calculation of size and power in large-scale testing problems. 
Consider a two component mixture model $f(z)=\pi_0f_0(z)+\pi_1f_1(z)$, in which $f_0$ and $f_1$ are densities of $z$-values under the null and non-null hypotheses; for each $1 \leq i \leq n $, $H_{0,i}$ is then true with probability $\pi_0$ and false with probability $\pi_1 = 1-\pi_0$. 
For $z$-values $Z_1, Z_2, \ldots, Z_n$ corresponding to the $H_{0,i}$, \citeauthor{Eb07} defined the $\Lfdr$ as 
\begin{equation*}
\Lfdr(z)= \mathbb{P}\left( H_{0,i} \text{ is true } \mid Z_i=z \right)=\dfrac{\pi_0f_0(z)}{f(z)}, 
\end{equation*}
and proposed to reject $H_{i,0}$ whenever $\Lfdr(Z_i)$ does not exceed a threshold  $\lambda >0$. 

As yet another alternative to the FDR, \citet{SC07} and \citet{Xj11} developed adaptive multiple testing procedures based on $z$-values and showed that their procedure is \emph{optimal} in the sense that it minimizes the empirical false negative rate, $\mFNR = \EE(T)/\EE(n-R)$, while controlling the empirical false discovery rate, $\mFDR = \EE(V)/\EE(R)$; see Table~\ref{tab:classification}. 
Let $F_0(z)$, $F_1(z)$ and $F(z)$ be cumulative density functions (cdf's) corresponding to $f_0(z)$, $f_1(z)$ and $f(z)$, respectively. Alternatively, $\mFDR$ can then be defined as the posterior probability of a case being null given that its $z$-value $Z_i$ is less than some cutoff $z$, 
\begin{equation*}
\mFDR(z)= \mathbb{P}(H_{0,i} \text{ is true } \mid Z_i \leq z )=\dfrac{\pi_0F_0(z)}{F(z)}. 
\end{equation*}
It can thus be seen that $\mFDR$ and $\Lfdr$ are analytically related: 
\begin{equation*}
\mFDR(z)= \int_{-\infty}^{z}  \Lfdr(u)f(u) du \Bigg/ \int_{-\infty}^z f(u)du =  \EE\left[\Lfdr(Z) \mid Z \leq z\right]. 
\end{equation*}

\citet{GW04} showed that for independent hypotheses, $\mFDR$ ($\mFNR$) and $\FDR$ ($\FNR$) are asymptotically equivalent, in the sense that $\mFDR = \FDR +\bigO(n^{-1/2})$. 
Thus, the procedure of \citeauthor{SC07} asymptotically controls the FDR. 
However, procedures based on $z$-values may not provide exact  FDR control at a given level $\alpha$. 

Despite significant progress in multiple hypothesis testing, existing approaches are mainly suitable for univariate hypotheses, i.e., for testing multiple univariate $p$- or $z$-values. 
Increasingly, however, multivariate, and potentially high-dimensional,  evidence is available for each hypothesis. 
A prime example of such applications arises in \emph{data integration} for high-throughput biology: As new high-throughput technologies emerge, biomedical scientists now collect multiple types of data on various aspects of cellular function, including DNA variants, copy number variation, DNA methylation, mRNA expression and abundances of proteins and metabolites. 
The Cancer Genome Atlas (TCGA) \citep{TCGA}, the ENCODE project \citep{ENCODE} and the Genotype-Tissue Expression (GTEx) project \citep{GTEx} are just a few examples of massive efforts to collect diverse high-throughput data in order to accelerate scientific discoveries. 
The hope, in these and other projects, is to delineate cellular functions and mechanisms of disease initiation and progression by integrating the evidence from diverse high-throughput data. 

A key step in ensuring the reproducibility of findings from testing many multivariate hypotheses is controlling the number of false discoveries. 
However, when $d > 1$, there is no unique way to rank multivariate hypotheses tests based on $p$-values, which is necessary for applying BH-type procedures. It is therefore not clear how to use BH-type procedures in order to control the FDR in multivariate settings. 
On the other hand, methods based on mixture model and density estimation can, in principal, be applied when $d > 1$. 
An early attempt to control the FDR in the multivariate setting was the proposal of \citet{Pa06}, who tried to generalize \citeauthor{Eb07}'s $\Lfdr$ as a function of multivariate $z$-values. They proposed to estimate $f_0$, $f_1$ and $\pi_0$ and then construct a local FDR function for rejecting $d=2$ dimensional hypotheses.  
However, \citeauthor{Eb07}'s method assumes that $\Lfdr(z)$ decreases as $|z|$ becomes larger. Under this assumption, $\FDR(z)$ is smaller than $\Lfdr(z)$, so controlling the $\Lfdr$ also controls the FDR. 
For instance, if $f_{(i)}$ is the standard normal density, then the family of densities $\{f_{(i)}(z-\mu):\mu \in \mathbb{R}\}$ has the monotone decreasing local FDR property (see, e.g., \citep{SC07} for more details). 
Therefore, in this case, controlling the $\Lfdr(z)$ at any level  guarantees that the FDR does not exceed that threshold. 
However, in many cases, for instance when $f_{(i)}$ have heavy tails or multiple modes, $\Lfdr(z)$ is no longer monotone decreasing. 
In such settings, the methods based on $\Lfdr$, e.g. \citep{Eb07} and \citep{Pa06}, are not guaranteed to control the FDR. 

\begin{table}[t]
\centering
\caption{Estimated FDR and FNR for method of \citet{SC07} (SC) and the oracle method in high dimensions. As the dimension $d$ increases, SC fails to control the FDR at the pre-specified level of $10\%$.}\label{tab:Cai}
\begin{adjustbox}{max width=\textwidth}
    \begin{tabular}{rrrrr}
  \hline
$d$ & SC FDR & Oracle FDR & SC FNR & Oracle FNR \\ 
  \hline
  2 & 0.11 & 0.10 & 0.13 & 0.13 \\ 
    3 & 0.14 & 0.10 & 0.12 & 0.12 \\ 
    4 & 0.20 & 0.10 & 0.09 & 0.06 \\ 
    5 & 0.40 & 0.10 & 0.09 & 0.03 \\ 
    6 & 0.52 & 0.10 & 0.07 & 0.02 \\ 
    7 & 0.66 & 0.10 & 0.08 & 0.02 \\ 
    8 & 0.73 & 0.10 & 0.08 & 0.07 \\ 
    9 & 0.76 & 0.10 & 0.04 & 0.00 \\ 
   10 & 0.79 & 0.10 & 0.04 & 0.00 \\ 
   11 & 0.80 & 0.10 & 0.01 & 0.01 \\ 
   12 & 0.80 & 0.10 & 0.01 & 0.00 \\ 
   13 & 0.80 & 0.10 & 0.00 & 0.00 \\ 
   14 & 0.80 & 0.10 & 0.01 & 0.00 \\ 
   15 & 0.80 & 0.10 & 0.05 & 0.00 \\ 
   \hline
    \end{tabular}
      \centering
 \end{adjustbox}
\end{table}

Another option for multivariate FDR control is the proposal of \citeauthor{SC07} (SC), which requires estimation of the multivariate densities, $f_0$ and $f_1$, as well as the proportion of null hypotheses $\pi_0$. 
Unfortunately, density estimation becomes increasingly difficult as the dimension $d$ increases. 
For instance, multivariate kernel density estimation is challenging when $d > 5$, the best possible (minimax) rate of mean-squared error of kernel density estimation is $\bigO(n^{-4/(4+d)})$ \citep{Pr08}. 
This bound underscores the ``curse of dimensionality'' when $d$ is large. 
As a result, the asymptotic FDR control of \citeauthor{SC07}'s method becomes invalid when $d$ is large. 
Table~\ref{tab:Cai} illustrates the inability of the SC method to control the FDR in multivariate settings in comparison to an oracle procedure, which assumes that $f_1$ and $\pi_0$ are known. 
Here, $n=10\,000$ z-values are generated from a mixture of Gaussians with $\pi_0=0.8$; the null hypotheses are independent standard normals, and the alternative hypotheses have a Gaussian density with a random covariance matrix and mean vector of length one. 
Since estimating $\pi_0$ in high dimensions is not straightforward (for instance the method of \citet{JC07} is not applicable), in this simulation, we have used the true value of $\pi_0 = 0.8$. The results clearly show that even for moderate dimensions, e.g. when $d = 3$ or 4, SC fails to control the FDR at the nominal level of $\alpha=0.10$.  

An alternative to FDR control for multivariate hypotheses is to combine the evidence from the $d$-variate hypotheses into a single summary measure. For instance, one can apply the Fisher's transformation for combining the $d$ $p$-values corresponding to each test of hypothesis to obtain a single $p$-value. BH-type methods can then be used to control the FDR using the resulting univariate $p$-value. However, the validity of such summaries often relies on strong assumptions. Moreover, as we will show in Section~\ref{sec:sim} such approaches can also be inefficient. 

In this paper, we propose a new procedure that overcomes the shortcomings of existing approaches for controlling the FDR in multivariate settings. 
The proposed procedure utilizes a stepwise rejection strategy to achieve exact FDR control, and is thus a BH-type procedure. However, our approach is based on test statistics, or $z$-values. To achieve FDR control for arbitrary $d$, at each step of the procedure, we use not the entire data, but the part corresponding to previously rejected hypotheses. 
In other words, once a hypothesis is rejected, we use its $z$-value to improve our estimate of the non-null distribution, and consequently the likelihood ratio statistics for other hypotheses. 
To show that the proposed procedure achieves exact control of the FDR, we first present a new proof for the BH step-down procedure. 
The new proof uses techniques from stochastic calculus and is amenable to higher dimensions. It also allows us to use random data-driven rejection regions. 
The key to controlling the FDR using this new procedure is that, in each step, we are not allowed to change our decision about previously rejected hypotheses, nor to use the $z$-values corresponding to hypotheses that are not yet rejected. Fortunately, since hypotheses are rejected according to their (estimated) likelihood ratio statistics, the previously rejected $z$-values contain the most information about the non-null distribution. We are thus able to achieve asymptotically optimal multiple hypothesis testing in high dimensions.  
Throughout the paper, we formulate our procedure and results in terms of more geometric notions of rejection regions rather than the ultimate univariate test statistic. However, these formulations become equivalent when the regions are level sets of $d$-dimensional test statistics.

The rest of the paper is organized as follows. In Section~\ref{sec:newpf}, we present our new proof of the BH step-down procedure. Using this result, in Section~\ref{sec:FDRinHD} we establish exact FDR control in higher dimensions. In Section~\ref{sec:oraclerule}, we first present an oracle decision rule for multivariate hypothesis testing assuming that the non-null distribution $f_1$ is known. The generalization of this approach to the settings where $f_1$ is estimated is presented Section~\ref{sec:approx}. The results of applying the proposed procedure to simulated and real data examples, as well as comparisons to existing approaches are presented in Sections~\ref{sec:sim} and \ref{sec:app}, respectively. We conclude the paper with a discussion in Section~\ref{sec:disc}.

\section{A New Proof for Step-Down BH Procedure}\label{sec:newpf}
Let $H_1,\ldots, H_n$ be $n$ null hypotheses and $p_1, \ldots, p_n$ be their corresponding $p$-values with $p_i \sim \mbox{Unif}[0,1]$  whenever the $i$th null hypothesis is true. We assume that $p_i$'s corresponding to null hypotheses are independent of each other and also independent of the $p$-values of non-null hypotheses.
Let  $p_{(1)}\leq \ldots \leq p_{(n)}$  be the ordered $p$-values, and  $H_{(1)}, \dots , H_{(n)}$ be their corresponding hypotheses. 

Consider the following empirical processes 
\begin{align*}
a_t &=\#\{\textnormal{null}\  p_i: p_i\leq t\}, \\
b_t &=\#\{ \textnormal{non-null}\  p_i: p_i\leq t\}, \\
r_t &=a_t+b_t.
\end{align*}
The false discovery proportion (FDP) and false discovery rate (FDR) can then be defined as 
\begin{align*}
\FDP(t) &=\dfrac{a_t}{r_t\vee 1}, \\
\FDR(t) &=\mathbb{E}\left[\dfrac{a_t}{r_t\vee 1}\right].
\end{align*}
The step-down BH procedure rejects the hypotheses $H_{(1)},\ldots, H_{(i_{sd})}$, where
\begin{equation}
i_{sd} = \max \{ i:p_{(i)}\leq \tau_{sd} \}, \text{ with }  
\tau_{sd} = \inf \left\{t:\  r_t\leq \dfrac{nt}{q}\right\}, 
\label{eq:BH}
\end{equation}
and accepts the rest.
Here $q$ is a user-specified ``error'' rate, or more precisely, the desired level of FDR control. 
\citet{BH95} showed that the step-down BH procedure controls the FDR at the level $q$. 

We next present an alternative proof for the step-down BH procedure, using tools from stochastic calculus. Using the new proof, in Section~\ref{sec:FDRinHD} we generalize the BH procedure to multivariate hypotheses. An important feature of the new proof is that it allows the rejection region to be random and to vary depending on the previously rejected hypotheses. This property will prove particularly useful in the development of FDR controlling procedures for multivariate hypothesis testing in Sections~\ref{sec:FDRinHD} and the proposed optimal procedure in Sections~\ref{sec:oraclerule}.

\begin{theorem}\label{thm1}
If the $p$-values corresponding to null hypotheses are independent of other $p$-values, then the step-down BH
algorithm \eqref{eq:BH}, denoted $BH_{sd}(qr_)$, controls the FDR at the level $q$, i.e.,
\begin{equation}\label{eq1}
\mathbb{E}\left[\dfrac{a_{\tau_{sd}}}{r_{\tau_{sd}}}\right]\leq \pi_{0}q\leq q. 
\end{equation}
\end{theorem}
\begin{proof}
By the definition of $\tau_{sd}$ in \eqref{eq:BH}, the statement \eqref{eq1} is equivalent to 
\begin{equation*}
\mathbb{E}\left[\dfrac{a_{\tau_{sd}}}{r_{\tau_{sd}}}\right]=\mathbb{E}\left[\dfrac{q}{n}\dfrac{a_{\tau_{sd}}}{\tau_{sd}}\right] \leq q. 
\end{equation*}
However, $r_t =a_t+b_t \leq \dfrac{nt}{q}$, and thus, $a_t\leq \dfrac{nt}{q}-b_t$.
Given that $b_t$ is independent of $a_t$, it suffices to show that $\mathbb{E}\left[\dfrac{a_{\tau_{sd}}}{r_{\tau_{sd}}} \mid b_t\right]\leq q$. We can thus assume, without loss of generality,  that $b_t$ is an arbitrary deterministic process. 

Now, let $\hat{a}_t=a_t-n_0 t$, where $n_0$ is the number of null hypotheses. Then, $\tau_{sd}$ can be redefined as 
\begin{equation*}
\tau_{sd}= \inf\left\{ t:\  \hat{a}_t\leq \dfrac{nt}{q}-b_t -n_0 t \right\}.
\end{equation*}
We can then write
\begin{equation*}
\mathbb{E}\left[\dfrac{a_{\tau_{sd}}}{r_{\tau_{sd}}}\right] = 
\mathbb{E}\left[\dfrac{q}{n}\dfrac{a_{\tau_{sd}}}{\tau_{sd}}\right] \leq \mathbb{E}\left[\dfrac{q}{n_0}\dfrac{a_{\tau_{sd}}}{\tau_{sd}}\right] \leq q.
\end{equation*}
Equivalently, 
\begin{equation*}
\mathbb{E}\left[\dfrac{1}{n_0}\dfrac{a_{\tau_{sd}}}{\tau_{sd}}\right] \leq 1, \text{ or, }
\mathbb{E}\left[\dfrac{\hat{a}_{\tau_{sd}}}{\tau_{sd}}\right] \leq 0.
\end{equation*}
The statement of the theorem now follows by the following property of the process $\hat{a}_t$, established in Lemma~\ref{lem}.
\end{proof}

\begin{lemma}\label{lem}
Let $\hat{a}_t$ be a process as defined above and $g$ be an arbitrary deterministic function. For a stopping time $\tau$ defined as 
\begin{center}
$\tau=\inf\left\{t: \hat{a}_t\leq g(t)\right\}$, 
\end{center}
we have
\begin{equation}\label{eq2}
\EE\left[\dfrac{\hat{a}_{\tau}}{\tau}\right] \leq 0. 
\end{equation}
\end{lemma}
\begin{proof}[Proof of Lemma~\ref{lem}]
We first prove that 
$\EE\left[\dfrac{\hat{a}_{\tau}}{\tau +\epsilon}\right] \leq 0$. 
Note that
\begin{equation}\label{eq3}
\dfrac{\hat{a}_{\tau}}{\tau +\epsilon}=\int_{0}^{\tau}\dfrac{1}{s +\epsilon}\,d \hat{a}_s - \int_{0}^{\tau}\dfrac{\hat{a}_s}{(s +\epsilon)^2}\,ds.
\end{equation}
Now, we claim that the process $\eta _t$ defined as
\begin{equation}\label{eq4}
d\eta _s =d\hat{a}_s +\dfrac{\hat{a}_s}{1-s}\ ds,
\end{equation}
is a martingale with respect to the filtration $\mathcal{F}_s = \sigma(a_u, u\leq s)$. 
To prove this claim, it suffices to show that  $\mathbb{E}\left[\eta_t - \eta_s \mid \mathcal{F}_s\right]=0$. But,
\begin{equation}\label{eq5}
\begin{split}
\mathbb{E}\left[\eta_t - \eta_s \mid \mathcal{F}_s\right] &= 
\mathbb{E}\left[\hat{a}_t - \hat{a}_s +\int_{s}^{t} \dfrac{\hat{a}_u}{1-u}\ du \mid \mathcal{F}_s\right]\\
 &= \mathbb{E}\left[\hat{a}_t -\hat{a}_s \mid \mathcal{F}_s\right] +\int_{s}^{t} \dfrac{\mathbb{E}\left[\hat{a}_u \mid \mathcal{F}_s\right]}{1-u}\ du.
 \end{split}
 \end{equation}
It is then easy to see that the process $\hat{a}_t$ satisfies 
\begin{equation*}
\mathbb{E}\left[\hat{a}_t \mid \mathcal{F}_s\right]=\dfrac{1-t}{1-s}\hat{a}_s,  
\end{equation*}
which is true because $a_t-a_s$ is independent of $\mathcal{F}_s$ and has binomial distribution $B\left(n_0 -a_s, \dfrac{t-s}{1-s}\right)$. 
The right-hand side of Equation~\eqref{eq5} can thus be written as
\[
\mathbb{E}[\hat{a}_t | \mathcal{F}_s] = \dfrac{s-t}{1-s}\hat{a}_s + \int_{s}^{t} \frac{\frac{1-u}{1-s}\hat{a}_s}{1-u} du=0.
\]
Now, substituting Equation~\eqref{eq4} into Equation~\eqref{eq3}, we get
\begin{eqnarray*}
\dfrac{\hat{a}_{\tau}}{\tau +\epsilon}&=&\int_{0}^{\tau}\dfrac{1}{s +\epsilon}\,d \eta _s - \int_{0}^{\tau}\dfrac{\hat{a}_s}{(s +\epsilon)(1-s)}\,ds- \int_{0}^{\tau}\dfrac{\hat{a}_s}{(s +\epsilon)^2}\,ds\\
&=&\int_{0}^{\tau}\dfrac{1}{s +\epsilon}\,d \eta _s - \int_{0}^{\tau}\dfrac{1+\epsilon}{(s +\epsilon)^2(1-s)}\hat{a}_s\,ds.
\end{eqnarray*} 
Taking expectation, we obtain:
\begin{equation}\label{eq:int}
\mathbb{E}\left[\dfrac{\hat{a}_{\tau}}{\tau +\epsilon}\right]=\mathbb{E}\left[\int_{0}^{\tau}\dfrac{1}{s +\epsilon}\,d \eta _s\right] -\mathbb{E}\left[ \int_{0}^{\tau}\dfrac{1+\epsilon}{(s +\epsilon)^2(1-s)}\hat{a}_s\,ds\right].
\end{equation}

The first term in Equation~\eqref{eq:int} is an integral with respect to a martingale and hence a martingale. Also, $\tau$ is a stopping time. So,  by the optional sampling theorem \cite{Ka13}, the expectation vanishes. Thus, 
\begin{equation*}
E\left[\dfrac{\hat{a}_{\tau}}{\tau +\epsilon}\right]=-E\left[ \int_{0}^{\tau}\dfrac{1+\epsilon}{(s +\epsilon)^2(1-s)}\hat{a}_s\,ds\right].
\end{equation*}
From this we have 
\begin{eqnarray*}
\mathbb{E}\left[ \int_{0}^{\tau}\dfrac{1+\epsilon}{(s +\epsilon)^2(1-s)}\hat{a}_s\,ds\right] 
&=& \mathbb{E}\left[ \int_{0}^{\infty}\dfrac{1+\epsilon}{(s +\epsilon)^2(1-s)}\hat{a}_s1_{\{s \leq \tau \} }\,ds\right],\\ 
&=& \int_{0}^{\infty}\dfrac{1+\epsilon}{(s +\epsilon)^2(1-s)}\mathbb{E}\left[\hat{a}_s1_{\{s \leq \tau \} }\right]\,ds, \\ 
&=& \int_{0}^{\infty}\dfrac{(1+\epsilon)P(\tau \leq s)}{(s +\epsilon)^2(1-s)}\mathbb{E}\left[\hat{a}_s | s \leq \tau\right]\,ds.
\end{eqnarray*}
However, $ t\leq \tau$ means that $\hat{a}_t\geq g(t)$ for $0 \leq s \leq t$ and hence $\hat{a}_s |\{t\leq \tau\}$ is stochastically larger than $\hat{a}_s$. Thus, $\mathbb{E}\left[\hat{a}_s \mid t\leq \tau\right] \geq \mathbb{E}\left[\hat{a}_s\right]=0$. To complete the proof it is sufficient to use Fatou's lemma:
\begin{equation*}
\mathbb{E}\left[\dfrac{\hat{a}_{\tau}}{\tau}\right] =
\mathbb{E}\left[\liminf _{\epsilon\to\ 0}\dfrac{\hat{a}_{\tau}}{\tau +\epsilon}\right] \leq  
\liminf _{\epsilon\to\ 0}\mathbb{E}\left[\dfrac{\hat{a}_{\tau}}{\tau +\epsilon}\right] \leq 0.
\end{equation*}
\end{proof}

\section{FDR Control for Multivariate $z$-Values}
\label{sec:FDRinHD}
In this section, we extend the step-down BH procedure for multivariate test statistics. The main difficulty in this case arises from the fact that in two, or higher, dimensions we have many choices for enlarging the rejection region. Thus, the original argument based on $p$-value rankings is not directly applicable. 
We will prove that the control of $\FDR$ remains valid, even when the method for enlarging the rejection region at each level is adapted to $z$-values rejected prior to that level. 
This finding can help improve the power of the test---i.e., reduce the FNR---since by observing previously rejected $z$-values, which likely correspond to non-null hypotheses, we can better estimate the non-null distribution and exploit this information to define the next rejection region more efficiently. An optimal algorithm based on this strategy is presented in Section~\ref{sec:oraclerule}. 

\subsection{Step-Down BH Procedures in Higher Dimensions}\label{sec:stepdownBH}
For a given domain $\mathcal{R} \subset \mathbb{R}^d$, we define the following empirical processes 
\begin{align*}
a(\mathcal{R}) &=\#\{\textnormal{null}\  z_i: z_i\in \mathcal{R}\},\\
b(\mathcal{R}) &=\#\{\textnormal{non}\ \textnormal{null}\  z_i: z_i \in \mathcal{R}\},\\
r(\mathcal{R}) &=a(\mathcal{R})+b(\mathcal{R}).
\end{align*}
Then,
\begin{align*}
\FDP(\mathcal{R}) &=\dfrac{a(\mathcal{R})}{r(\mathcal{R}) \vee 1},\\
\FDR(\mathcal{R}) &=\mathbb{E}\left[\dfrac{a(\mathcal{R})}{r(\mathcal{R}) \vee 1}\right].
\end{align*}
\begin{figure}[t]
  \centering
  \includegraphics[trim = 25mm 25mm 20mm 30mm, clip, width=7cm]{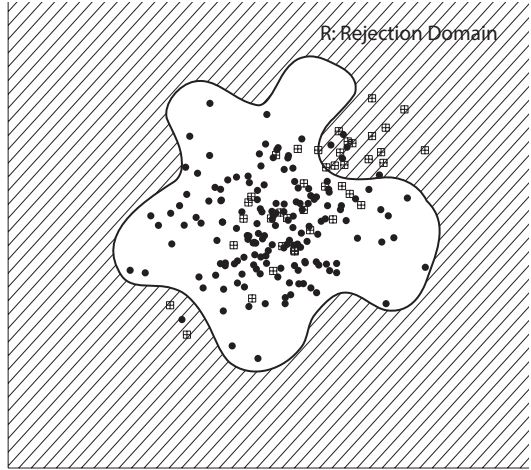}
  \caption{Illustration of adaptive rejection domains; the shaded area $\mathcal{R}$ represents the rejection domain, circles represent null $z$-values and squares represent alternative hypotheses $z$-values.}
\end{figure}

We now present a formal definition of the new step-down procedure. 
\begin{defn}\label{defn:NR}
Let $\mathcal{F}_{\mathcal{R}}=\sigma\left(1_{\{z_i \in B \}}:B \subset {\mathcal R},\ i=1,\ldots,n\right) $ for all $\mathcal{R} \subset \mathbb{R}^d$. An adaptive step-down BH algorithm is a 
family of increasing stopping sets $\{\mathcal{R}_{t}\}_{0 \leq t \leq 1}$ satisfying the following conditions 
\begin{enumerate}
\item for $s\leq t$ we have $\mathcal{R}_{s} \subset \mathcal{R}_{t}$; 
\item $F_0(\mathcal{R}_{t})=t$;  
\item $\forall A \subset \mathbb{R}^d \ ,\{\omega:\mathcal{R}_t (\omega) \subset A \} \in \mathcal{F}_A$. 
\end{enumerate}
\end{defn}
Definition~\ref{defn:NR} requires that the rejection regions form a family of increasing stopping sets developed based on previously rejected hypotheses. 
For the family of increasing stopping sets in Definition~\ref{defn:NR}, our proposed generalized step-down BH method reject all $z$-value in $\mathcal{R}_{\tau}$, where
\begin{equation}\label{eq:tau}
 \tau = \inf \left\{t :  r(\mathcal{R}_t)\leq \dfrac{nt}{q} \right\}.
\end{equation}  
We next show that, for any error rate $q$, the proposed generalized BH procedure with the threshold defined in \eqref{eq:tau} controls the FDR at the level $q$.  
\begin{figure}[t]
  \centering
  \includegraphics[trim = 10mm 10mm 10mm 20mm, clip,width=6cm]{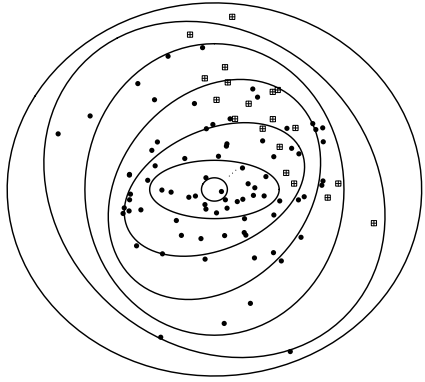}
  \caption{Illustration of the family of increasing nested rejection regions $(\mathcal{R}_t)_{0 \leq t \leq 1}$; the rejection regions are constructed based on previously rejected $z$-values and are not restricted to have similar shapes.}
\end{figure}

\begin{theorem}\label{th3}
If the $z$-values corresponding to null hypotheses are independent of other $z$-values, then the generalized step-down BH algorithm controls the  false discovery rate at $q$,
\begin{equation}\label{eq10}
\mathbb{E} \left[\dfrac{a(\mathcal{R}_{\tau})}{r(\mathcal{R}_{\tau})}\right] \leq \pi_{0}q \leq q.
\end{equation}
\end{theorem}
\begin{proof}
First note that $\tau$ is a stopping time with respect to the filtration $\mathcal{F}_t = \sigma(1_{\{ z_i \in \mathcal{R}_{s} \}}, s\leq t) $. Similar to Section~\ref{sec:newpf}, we can rewrite $\tau$ as
 \begin{eqnarray*}
 \tau = \inf \Big\{t : a(\mathcal{R}_t)\leq  \underbrace {\dfrac{nt}{q} -b(\mathcal{R}_t)}_{g(t)}  \Big\}.
 \end{eqnarray*}  
If we define $\hat{a}(\mathcal{R}_t)=a(\mathcal{R}_t)-n_0 F_0(t)$,  then the assertion of Theorem~\ref{th3} is equivalent to
\begin{equation*}
\mathbb{E}\left[\dfrac{a(\mathcal{R}_\tau)}{r(\tau)}\right] =
\mathbb{E}\left[\dfrac{q}{n}\dfrac {a(\mathcal{R}_\tau)}{F_0(\tau)}\right] \leq 
\mathbb{E}\left[\dfrac{q}{n_0}\dfrac {a(\mathcal{R}_\tau)}{F_0(\tau)}\right] \leq q, 
\end{equation*}
which is, in turn, equivalent to 
\begin{equation*}
\mathbb{E}\left[\dfrac{1}{n_0}\dfrac {a(\mathcal{R}_\tau)}{F_0(\tau)}\right] \leq 1 
\Longleftrightarrow 
\mathbb{E}\left[\dfrac {\hat{a}(\mathcal{R}_\tau)}{F_0(\tau)}\right] \leq 0
\Longleftrightarrow 
\mathbb{E}\left[\dfrac {\hat{a}(\mathcal{R}_\tau)}{\tau}\right] \leq 0.
\end{equation*}
Now, since the null $z$-values are independent  and $\mathcal{R}_s$ is a stopping set,  $a(\mathcal{R}_t)-a(\mathcal{R}_s)$ has a binomial distribution $B\left(n_0 -a(\mathcal{R}_s), \dfrac{t-s}{1-s}\right)$. 
Thus, 
 \begin{equation*}
\mathbb{E}\left[a(\mathcal{R}_t)| \mathcal{F}_s\right]=a(\mathcal{R}_s) +(n_0 -a(\mathcal{R}_s)) \dfrac{t-s}{1-s}=\dfrac{1-t}{1-s}a(\mathcal{R}_s) +\dfrac{t-s}{1-s}n_0. 
\end{equation*}
Given the above equation, and as in the proof of Lemma~\ref{lem}, we conclude that
\begin{equation}\label{eq11}
\eta _t = a(\mathcal{R}_t)-\int _{0}^{t} \dfrac{n_0 -a(\mathcal{R}_s)}{1-s}\ ds
\end{equation}
is an $\mathcal{F}_s$-martingale. 
The rest of the proof is exactly the same as the proof of Theorem~\ref{thm1}. 
\end{proof} 
 
\section{Oracle Decision Rule for Multivariate Hypothesis Testing}\label{sec:oraclerule}
In this section, we construct an oracle decision rule based on \emph{nested rejection regions} for multiple testing in multivariate settings. Inspired by the compound decision rule framework of \cite{SC07}, we first develop an optimal oracle procedure in Theorem~\ref{th4}. 
Then, in Section~\ref{sec:approx}, we propose a nested rejection region algorithm that approximates the oracle rejection rule; we show that our algorithm is asymptotically optimal under weak conditions. 

Consider $n$ hypotheses $\mathrm{H}_{i}, i=1, \ldots, n$. Denote by $\theta_{i}$ the significance indicator for the $i^\text{th}$ hypotheses:  $\theta_{i}$ is $0$ if the $i^{th}$ null hypothesis is true and 1 otherwise. 
Throughout this section, we assume that $\theta_{1}$,\ldots,$\theta_{n}$ are independent and identically distributed Bernoulli variables with success probability $\pi_1 = 1-\pi_0$. The proposed procedure can be generalized to the setting of Markov random fields, as in \citet{SC09}. However, such extensions are beyond the scope of the current manuscript and is left to future research.   
Let $Z=(Z_{1},\ldots,\ Z_{n})$ be an independent sequence of test statistics for the $n$ hypotheses $H_i$ with the following density,
 \begin{equation}\label{eq:mixture}
 	f(Z_i =z|\theta_i = k) \sim f_k(z) \  \text{ for } k = 0, 1 \text{ and }  i=1,\ldots,n.
 \end{equation}
Our goal is to construct a multiple testing procedure, based on a realization $z = (z_{1},\ \ldots,\ z_{m})$, which achieves the maximum number of rejected hypothesis, while controlling the FDR at a pre-specified level $q$. In the next theorem, we assume that the value of $\pi_0$ and density function $f_1$ are known. 
Under these assumptions, part 1 of Theorem~\ref{th4} characterizes an optimal set $S$ in terms of {\Lfdr} with optimal power, while controlling the FDR at the level $q$. 
Parts 2 and 3 then establish the existence of such a set, whereas part 4 shows that the set $S$ can be obtained by a family of increasing nested stopping sets, defined in Section~\ref{sec:stepdownBH}. 
The last part of the theorem provides a specific recipe for constructing optimal rejection region $S$. Specifically, it shows that the optimal set $S$ can be constructed by applying the BH procedure to the nested regions constructed based on parts 2-4. 
 
 \begin{theorem}\label{th4}
The following hold for a multiple testing problem with $n$ independent hypotheses represented by $d$-dimensional test statistics $z_i, i=1, \ldots, n$ from the mixture distribution \eqref{eq:mixture}.
\begin{enumerate}
\item For a given $q$ if there exists an $l$ such that the set $S=\{x:\Lfdr(x)<l \}$  satisfies
 	\begin{equation*}
			\pi_0  \int_{S} f_0(x)dx = q\int_{S} f(x)dx, 
 	\end{equation*}
 	then for any other rejection set $T$ such that $\int_{T}\pi_0 f_0(x)dx \leq q\int_{T} f(x)dx$,
 	\begin{equation*}
 		\int_{T}f_1(x)dx\leq\int_{S}f_1(x)dx.
 	\end{equation*}
\item If $\inf(\Lfdr(x))>q$, then for any set $U$ with positive Lebesgue measure $\mu$, 
 	\begin{equation*}
 		\int_{U}\pi_0 f_0(x)dx > q\int_{U} f(x)dx. 
 	\end{equation*}
\item If $f$ is an analytic function and there exists an $x$ such that $\Lfdr(x)<q$, then the optimal set $S$ described in 1 exists. 
\item Under the conditions of part 3, there exists an increasing sequence of nested stopping sets $\{\Omega_t\}_{t=0}^{1}$ that satisfy the requirements of Definition~\ref{defn:NR}, and for which $\Omega_t=\{x:\Lfdr(x)< l_t \}$. 
\item Let $\tau^{ABH}_n$ be the stopping index resulting from applying the adaptive BH procedure of Section~\ref{sec:FDRinHD}, 
i.e., 
	\[
	\tau^{ABH}_n= \arg\min_{t} \left\{ q\leq \dfrac{\pi_0 F_0(\Omega_t)}{\hat{F}(\Omega_t)} \right\}. 
	\]
	Further, let $\tau$ be the index from part 4 such $S=\Omega_\tau$.  Then, as $n {\longrightarrow}\infty$, $\tau^{ABH}_n \overset{p}{\longrightarrow} \tau$. 
\end{enumerate}
\end{theorem}
\begin{proof}[Proof of Theorem~\ref{th4}] $\,$
\begin{enumerate} 

\item The proof of part 1 is similar to proof of Theorem~2.1 in  \citet{JZ14} with slight modifications. 

	First note that since $\Lfdr(x)$ is a monotone decreasing function of the likelihood ratio, $h(x)=\frac{f_{1}(x)}{f_{0}(x)}$. Thus, the set $S$ can be obtained by thresholding $h(x)$ at some level $l'$. 	
	Now, for any set $T,$ the definition of $S$ implies that
	\begin{equation*}
		(1_{S}-1_{T})(f_{1}(x)-l'f_{0}(x))\geq 0.
	\end{equation*}
	Let $\alpha(S)=\int_{S}f_{0}$ and $\beta(S)=\int_{S}f_{1}$ be the type I error and the power of the test based on the rejection region $S$. Define $\alpha(T)$ and $\beta(T)$ similarly. Then 
	\begin{equation}\label{eq41}
		0 \leq\int(1_{S}-1_{T})(f_{1}(x)-l'f_{0}(x))dx=\beta(S)-\beta(T)+l'(\alpha(T)-\alpha(S)).
	\end{equation}
	But for any set $T$ satisfying the requirements of the theorem, 
	\begin{equation*}
		\pi_0\int_{T}f_{0}(x)-q\int_{T}f(x)dx\leq 0=\pi_0\int_{S}f_{0}(x)- q\displaystyle \int_{S}f(x)dx. 
	\end{equation*}
	Equivalently,
	$(1-q)\pi_0(\alpha(T)-\alpha(S))\leq q(1-\pi_0)(\beta(T)-\beta(S))$ or $q'(\alpha(T)-\alpha(S))\leq\beta(T)-\beta(S)$,
	where $q' = \frac{(1-q)\pi_0}{q(1-\pi_0)}$. 
	Combining this with Equation~\eqref{eq41}, we get
	\begin{equation*}
		(l'-q')(\beta(T)-\beta(S))>0.
	\end{equation*}
	Since $\pi_0\int_{S}f_{0}(x)dx=q\int_{S}f(x)dx$, then $q' \int_{S}f_{0}(x)=\int_{S}f_{1}(x)<l'\int_{S}f_{0}(x)$. Consequently, $l'-q'>0$, which implies that $\beta(S)\geq\beta(T)$. This completes the proof.
\item If $\inf(\Lfdr(x))>q$ then for any $x \in U$ we have $\pi_0f_0(x)-qf(x)>0$. Now, since $U$ has a positive measure and 
	\begin{equation*}
		U=\bigcup_{k \geq 1} U_k, \quad  \mbox{ where } \quad U_k = \left\{x \in U: \pi_0f_0(x)-qf(x) > \frac{1}{k}\right\},
	\end{equation*}
	there exists some $k$ for which $\mu(U_k)$ is positive.  
	We then have 
	\[
	\int_{U_k} \pi_0f_0(x)-qf(x) \, dx \geq \int_{U} \pi_0f_0(x)-qf(x) \, dx > \int_{U_k} \frac{1}{k} \, d\mu = \dfrac{\mu(U_k)}{k} > 0, 
	\]
	as desired. 
\item Let $\Omega_\lambda :=\{x:Lfdr(x)<\lambda\}$. We assert that the  function $g:\mathbb{R} \longrightarrow \mathbb{R},\  g(\lambda):=F_0(\Omega_\lambda)$ is continuous. It is sufficient to prove that for any monotone sequence $\lambda_k$ converging to $\lambda$, $g(\lambda_k) \longrightarrow g(\lambda)$. Assume  $g(\lambda_k) \nrightarrow g(\lambda)$. 
We consider two cases. 
First, assume $\lambda_{k} \nearrow \lambda$. Then, $\bigcup\limits_{k=1}^{\infty} \Omega_{\lambda_{k}} \setminus\Omega_{\lambda} =\emptyset$. Hence $F_0(\Omega_\lambda\setminus\Omega_{\lambda_{n}})\longrightarrow 0$, implying that $g(\lambda_n) \longrightarrow g(\lambda)$. Next, suppose $\lambda_k  \searrow \lambda$. Then, $\bigcap\limits_{i=1}^{\infty} \Omega_{\lambda_{n}}\setminus\Omega_{\lambda}=\Omega^\lambda$, where  $\Omega^\lambda :=\{x:\Lfdr(x)=\lambda\}$. Therefore, if $g(\lambda_k) \nrightarrow g(\lambda)$ then $F_0(\Omega^\lambda) \neq 0$. But each non-zero empty set has a dense point. Therefore, there exists a sequence $(x_j)_{j=0}^{\infty} \in \Omega^\lambda$, such that $x_j \longrightarrow x$. We next show that all derivatives of $\Lfdr$ at point $x$ are zero. First, note that 
	\begin{equation}\label{eq:lim}
		\frac{d}{dx}\Lfdr(x) = \lim_{x_j \to x} \dfrac{\Lfdr(x_j)-\Lfdr(x)}{x_{j}-x} =\dfrac{\lambda - \lambda}{x_{j}-x} = 0.
	\end{equation} 
	Thus, \eqref{eq:lim} implies that first derivative of $\Lfdr$ at $x$ is zero. Higher derivatives can also be computed using the same formula, based only on the values of the function. Thus, all derivatives of $\Lfdr$ are zero. However, by the assumption, $\Lfdr$ is an analytic function. Now, by the Taylor expansion of $\Lfdr$ at point $x$, we obtain that $\Lfdr$ is constant, which is a contradiction. Thus, $g$ is continuous. A similar argument shows that $g(\lambda):=F(\Omega_\lambda)$ is also continuous. 
	
	We next characterize the optimal set $S$. Let $g(\lambda)=\pi_0\int_{\Omega_\lambda}f_{0}(x)dx-q\int_{\Omega_\lambda}dF(x)$ . Then, $g(\lambda)$ is continuous. Moreover, because $\inf(\Lfdr(x))>q$, if $\lambda <q$ then $\forall x \in \Omega_\lambda$, 
	\[
		\pi_0 f_0(x)<\lambda f(x) \Rightarrow \pi_0 F_0(\Omega_\lambda)<\lambda F(\Omega_\lambda) < qF(\Omega_\lambda) \Rightarrow g(\lambda)<0.
	\]
	Suppose, without loss of generality, that $\pi_0 >q$. (If $\pi_0 <q$,  we can reject all of the hypotheses.)  
	Now, let $\alpha = \sup(\Lfdr(x))$. Then, $g(\alpha)=\pi_0 -q >0$. Thus, by the mean value theorem, there exists some $\mu$ such that $\pi F_0(\Omega_\tau)=qF(\Omega_\tau)$. This completes the proof. 
\item Note that $F_0(\Omega_\lambda)$ is a continuous and increasing function with respect to $\lambda$. Let $\Omega_0=\emptyset$. We can then find a function $\alpha: [0,1] \longrightarrow \mathbb{R}^+$ such that $F_0(\Omega_{\alpha(t)})=t$. Setting $\mathcal{R}_t=\Omega_{\alpha(t)}$ then gives a family of nested rejection region $\{\mathcal{R}_t\}_{t=0}^{1}$ satisfying the requirements of Definition~\ref{defn:NR}. 
\item Let $\hat{F_n}(\Omega_\lambda)=\frac{1}{n} \sum\limits^{n}_{i=1} I(z_i \in \Omega_\lambda)$ be the  empirical distribution function. Then, by Dvoretzky-Kiefer-Wolfowitz inequality (DKW)~\citep{DKW56}, 
\[
	\PP\left(\sup\limits_{t\in[0, 1]}|\hat{F_n}(\Omega_\lambda)-F(\Omega_\lambda)|> \epsilon \right) \leq e^{-2n\epsilon^2}. 
\]
Let $\epsilon_n = n^{-\frac{1}{3}}$. Then, for sufficiently large $n$, with high probability $F(\Omega_\lambda)-\epsilon_n<\hat{F_n}(\Omega_\lambda)<F(\Omega_\lambda)+\epsilon_n$. Let $\tau_n^{\pm}$ be the solution of equation $\pi_0F_0(\Omega_{t})=qF(\Omega_t)\pm \epsilon_n$. Because the function $g(t)=\dfrac{\pi_0F_0(\Omega_{t})}{qF(\Omega_t)}$ is monotone increasing with respect to $t$, we have 
\[
\tau_n^{-}<\tau_n^{ABH}<\tau_n^{+}.
\] 
We assert that $\tau_n^{+}$ converges to $\tau$. Since $\tau_n^{+} \in [0,1]$ if the sequence diverges, then there exists at least two subsequences $\{\tau^{+}_{1,n}\}_{n=1}^{\infty}$ and $\{\tau^{+}_{2,n}\}_{n=1}^{\infty}$ converging to, say, $\tau^+_1$ and $\tau^+_2$, respectively. But $g(t)$ is one to one   and $g(\tau^+_1)=g(\tau^+_2)=g(\tau)=q$. Thus, $\tau^+_1=\tau^+_2 =\tau$, which implies that $\tau^+ \rightarrow \tau$. By a similar argument, $\tau^- \rightarrow \tau$. Thus, $\tau_n^{ABH} \rightarrow \tau$, as desired. 
\end{enumerate}
\end{proof}
Theorem~\ref{th4} shows that an optimal multiple testing strategy can be obtained using a family of nested rejection regions. Importantly, the theorem establishes that instead of searching over all possible rejection regions in $\mathbb{R}^d$, one only needs to search over the collection of sets $\Omega_{t}$ in order to find optimal rejection set. However, Theorem~\ref{th4} assumes that $f_1$ and $\pi_0$ are known. We relax this assumption in the next section. 

\subsection{Approximation of the Oracle Procedure}\label{sec:approx}
In Theorem~\ref{th4}, we proved that given the density $f_1$ of alternative hypotheses and the proportion of null hypotheses $\pi_0$, the adaptive BH algorithm of Section~\ref{sec:FDRinHD} converges to the optimal rejection set, as $n \to \infty$. However, in practice, $f_1$ and $\pi_0$ are often unknown. The oracle procedure of Theorem~\ref{th4} is hence not directly applicable. 
In Algorithm~\ref{alg1} we propose an approximation to the oracle procedure of Theorem~\ref{th4} without assuming that $f_1$ and $\pi_0$ are known. 
For this algorithm, we have:
\begin{theorem}\label{th5}
	Let $\hat{\Omega}$ be the rejection region obtained from Algorithm~\ref{alg1} for an analytic density function $f$. Then as $n \to \infty$, $\hat{\Omega}$ converges in probability to some $\Omega_\tau \in \{\Omega_t\}_{t=0}^{1}$. 
\end{theorem}

Before presenting the proof of Theorem~\ref{th5}, we comment on two aspects of Algorithm~\ref{alg1}, which distinguish it from the oracle procedure of Theorem~\ref{th4}. 
First, in this algorithm $f_1$ is assumed to be unknown and is estimated from data. 
The algorithm builds \emph{nested rejection regions} (NR) by enlarging the rejection region in each step until the stopping criterion is met. 
Ideally $f_1$ should be estimated from all data points. In that case, a natural solution would be to estimate $f_1$, or alternatively $f$, using a kernel density estimator. However, the construction of adaptive rejection regions in Definition~\ref{defn:NR} prevents us from using all $z$-values; instead, we are restricted to using only  the rejected $z$-values at each step of the algorithm. 
The algorithm thus needs an initial rejection region. To this end, we propose to start by rejecting hypotheses with far enough $z$-values, e.g. we choose the initial rejection domain to be the complement of a big ball centered at the origin that controls the FDR at a level $q' \ll q$. The control of FDR by such a rejection region then follows from the validity of the original BH procedure. 
As the algorithm continues, the rejection region is refined by obtaining a more accurate estimate of $f$. 
The main theorem in this section, Theorem~\ref{th5}, shows that if we estimate $f$ by applying a kernel density estimator to $z$-values in the rejection domain, and use $\hat{f}/f_0$ as the test statistic for extending the rejection domain, then the final rejection region from Algorithm~\ref{alg1} is asymptotically optimal.

\begin{algorithm}[t!]
{\singlespacing
\caption{The Nested Rejection Region (NR) Algorithm \\ for Step-Down FDR Control in Multivariate Hypothesis Testing}\label{alg1}
\begin{itemize}
\item[] \textbf{Initialization}: 
	\begin{itemize}
	\item Set $\Omega = \Omega_0$ the initial rejection region $\Omega_0 = \{ z : \lVert z \rVert > \lambda \} $ such that $F_0(\Omega_0) = q' \ll q$. 
	\item Reject z-value in $\Omega$ and let $\mathcal{R} = \{ z_i : z_i \in \Omega \}$. 
	\item Set $\FDR = q'$.
	\end{itemize}
\end{itemize}

\begin{enumerate}
\item[] \textbf{while} $\FDR \leq q$
	\begin{enumerate}
	\item[1.] Define $\hat{f}_{\Omega}(y)=\dfrac{1}{nh} \sum_{z_i \in \Omega}\omega((y-z_i)/h).$
	\item[2.] Reject  $z^*$ and add to $\mathcal{R}$ where $z^*= \arg\max_{z_j} \left\{\frac{\hat{f}_{\Omega}(z_j)}{f_0(z_i)} : z_j \notin \mathcal{R} \right\}.$
	\item[3.] Set $\Omega_{\mbox{\scriptsize{search}}} =  \left\{z: \frac{\hat{f}_{\Omega}(z)}{f_0(z)} > \frac{\hat{f}_{\Omega}(z^*)}{f_0(z^*)} \right\}.$
	\item[4.] Set $\Omega = \Omega \cup \Omega_{\mbox{\scriptsize{search}}}.$
	\item[5.] Set $\mathrm{FDR} = \dfrac{F_0(\Omega)}{n^{-1}\left\vert{\mathcal{R}}\right\vert}.$
	\end{enumerate}
	\item[] \textbf{end while}
\end{enumerate}
}
\end{algorithm}

The second distinction between Algorithm~\ref{alg1} and the oracle procedure of Theorem~\ref{th4} concerns the knowledge of the proportion of null hypotheses $\pi_0$. 
In Theorem~\ref{th4} the optimal rejection region was constructed by assuming that $\pi_0$ is known. As a result, the optimal rejection region $S$ in Theorem~\ref{th4} controls $\mFDR$ at the exact level $q$. However, $\pi_0$ is often unavailable. 
The procedure of Algorithm~\ref{alg1} thus follows the conservative BH approach and (asymptotically) controls $\mFDR$ at the level $\pi_0 q$. 

Theorem~\ref{th5} shows that the rejection set obtained from Algorithm~\ref{alg1} approximates the optimal rejection region in Theorem~\ref{th4}, but with $\mFDR$ controlled at level $\pi_0 q$. 
More specifically, Theorem~\ref{th5} shows that the rejection region from Algorithm~\ref{alg1} converges to $\Omega_\tau \in \{ \Omega_\lambda : 0 \leq \lambda \leq 1 \}$ introduced in Theorem~\ref{th4}. 
Each member of this one-parameter family is uniquely characterized by the control level of $\mFDR$. For the rejection region $\hat\Omega$ selected by Algorithm~\ref{alg1}, we have 
\begin{equation*}
\lim\limits_{n \to \infty} {\dfrac{F_0(\hat\Omega)}{n^{-1}\left\vert{\mathcal{R}}\right\vert}} = \pi_0 q.  
\end{equation*}
However, Theorem~\ref{th5} states that, as $n\to \infty$,
\[
\dfrac{F_0(\hat\Omega)}{n^{-1}\left\vert{\mathcal{R}}\right\vert} \longrightarrow\dfrac{F_0(\Omega_\tau)}{F(\Omega_\tau)}, 
\]
for some $\tau$. Thus, $\Omega_\tau$ satisfies the requirements of Theorem~\ref{th4}, and in particular, satisfies the requirement of part 1 of Theorem~\ref{th4} at the (conservative) level $\pi_0 q$. This implies that $\Omega_\tau$ is optimal for FDR control in multivariate hypotheses at the level $\pi_0 q$, which in turn, establishes the (asymptotic) optimality of the rejection region from Algorithm~\ref{alg1}. 

We next give a proof of Theorem~\ref{th5}. 
\begin{proof}[Proof of Theorem~\ref{th5}]
	Suppose there exits an initial rejecting set $\mathcal{R}_0$ containing $z$-values for alternative hypotheses, i.e., $\{ z_i: \frac{f(z_i)}{f_0(z_i)}>c_0 \gg 1 \} \subset \mathcal{R}_0$. Let $z_{(1)}, z_{(2)}, \dots , z_{(i),\dots}$ be sorted $z$-values according to the step they are rejected in Algorithm~\ref{alg1}.
		
	For each $i$ define $\mathcal{R}_i$ to be the rejection region of the $i^{th}$ step containing $z$-value $z_{(1)}, z_{(2)}, \ldots, z_{(i)}$. For this region, define the estimated density function $\hat{f}_{\mathcal{R}_i}(y) = (ih)^{-1} \sum_{z_k \in \mathcal{R}_i}\omega((y-z_k)/h)$, where $\omega$ is an arbitrary kernel. 
	Let $P_i ^{(n)}= \min \left\{\frac{f(z_j)}{f_0(z_j)}: z_j \in \mathcal{R}_i \right\}$ and $Q_i^{(n)} = \max\left\{\frac{f(z_j)}{f_0(z_j)}: z_j \not\in \mathcal{R}_i \right\}$. If $\mathcal{R}_i$ does not converge to some $\Omega_\lambda = \left\{ z: \frac{f(z)}{f_0(z)} > \lambda \right\}$, then for any $n$ there exists an $\epsilon$ such that $Q_i^{(n)}-P_i^{(n)} > \epsilon$. We assert that the probability of this event converges to zero for all $i$, i.e., $\mathbb{P}\left(\max_i \{Q_i^{(n)}-P_i^{(n)} \} > \epsilon\right)  \longrightarrow 0  \ \ as \  n \rightarrow \infty $.  
	
	For each $z_0$ value, define $\Omega_{\lambda_{z_0}} = \left\{ z: \frac{f(z)}{f_0(z)} > \frac{f(z_0)}{f_0(z_0)} \right\}$ and let
	\begin{equation*}
		\tilde{f}_{z_0}(y) = \dfrac{1}{|\Omega_{\lambda_{z_0}}|h} \sum_{z_k \in \Omega_{\lambda_{z_0} }}  \omega((y-z_k)/h). 
	\end{equation*}
	Let  $i^* ={\mathrm{argmax} }_i \left\{Q_i^{(n)}-P_i^{(n)} \right\}$. Define $P = P_{i^*}^{(n)}$, $Q = Q_{i^*}^{(n)}$ and let $z_P$ and   $z_Q$ be the corresponding $z$-values. For these two $z$-values we have
	\begin{equation*}
		\dfrac{\tilde{f}_{z_Q}(z_Q)}{f_0(z_Q)} \geq \dfrac{\hat{f}_{z_Q}(z_Q)}{f_0(z_Q)} \geq \dfrac{\hat{f}_{z_P}(z_P)}{f_0(z_P)} \geq \dfrac{\tilde{f}_{z_P}(z_P)}{f_0(z_P)}, 
	\end{equation*}
	where $\tilde{f}_{z_Q}$ is defined similarly as $\tilde{f}_{z_0}$. 
	Therefore, 
	\begin{equation*}
		\left\{ \dfrac{\hat{f}_{z_P}(z_P)}{f_0(z_p)}-\dfrac{\hat{f}_{z_Q}(z_Q)}{f_0(z_Q)}   > \epsilon \right\} \subseteq  
		\left\{  \dfrac{\tilde{f}_{z_P}(z_P)}{f_0(z_p)}-\dfrac{\tilde{f}_{z_Q}(z_Q)}{f_0(z_Q)} > \epsilon \right\} \subseteq   
		\left\{  \dfrac{\tilde{f}_{z_P}(z_P)}{f_0(z_p)} -\dfrac{\tilde{f}_{z_Q}(z_Q)}{f_0(z_Q)}> 0 \right\}.
	\end{equation*}
	Let $\Omega_P = \left\{ z:  \frac{f(z)}{f_0(z)} >\frac{f(z_P)}{f_0(z_P)} \right\}$ and define $\Omega_Q$ similarly. We can write 
	\begin{equation*}
		\tilde{f}_{z_P}(x) = \dfrac{1}{|\Omega_P|h} \sum_{z \in \Omega_P} \omega((x-z)/h) =
		\frac{1}{|\Omega_P|} \sum_{i =1}^{n} h^{-1}\omega((x-z_i)/h)\mathbb{I}_{ \{ z_i \in \Omega_P\}},
	\end{equation*}
	where random variables $h^{-1}\omega((x-z_i)/h)\mathbb{I}_{ \{ z_i \in \Omega_P\}}$ are i.i.d. To complete the proof, we need the following lemma. 
	
\begin{lemma}\label{lemma4thm5}
	The expected value of $\tilde{f}_{z_P}(x)$ at the point $z_P$ converges to $\frac{1}{2} f(z_P)$. 
\end{lemma} 
\begin{proof}[Proof of Lemma~\ref{lemma4thm5}]
	At the point $z_P$, consider the region $\Omega_P = \left\{ z:  \frac{f(z)}{f_0(z)} >\frac{f(z_P)}{f_0(z_P)} \right\}$. Then, 
	\begin{align*}
		\mathbb{E}\left[\tilde{f}_{z_P}(x)\right] 
		&= \mathbb{E}\left[\frac{1}{|\Omega_P|} \sum_{z_i \in \Omega_P} \frac{1}{h} \omega((x-z_i)/h)\mathbb{I}_{ \{ z_i \in \Omega_P\}}\right] \\
		&= \mathbb{E}\left[\frac{1}{h} \omega((x-z_i)/h)\mathbb{I}_{ \{ z_i \in \Omega_P\}}\right] = \int_{\Omega_P} \omega(u)f(z_P+hu)du.
	\end{align*}
	For an arbitrary kernel $\omega$, we can choose a $\lambda > 0$ such that $\omega(u) < \frac{\epsilon}{4}$ for $\lVert u \rVert >\lambda$. Then, by the Taylor expansion of the analytic density function $f$, 
	\begin{align} \label{int1}
		\int_{\Omega_P} & \omega(u)f(z_P+hu)du \nonumber \\
		&= \int_{\Omega_P :\lVert u \rVert \leq \lambda } \omega(u)f(z_P+hu)du +\int_{\Omega_P:\lVert u \rVert > \lambda } \omega(u)f(z_P+hu)du \nonumber \\ 
		&<  \int_{\Omega_P :\lVert u \rVert \leq \lambda } \omega(u)f(z_P+hu)du + \frac{\epsilon}{4}. 
	\end{align}
	Let $\bar{\omega}(u) = \omega(u)\mathbb{I}_{\{\|u\| < \lambda\}}$ be the truncated kernel constructed above. 
	By this construction, $\Omega_{P: \lVert u \rVert \leq \lambda}$ is within the ball $\mathcal{B}_{h\lambda}(z_P)$ centered at $z_P$ with radius $h\lambda$. 
	
	Let $S_P = \left\{ z: \frac{f(z)}{f_0(z)} = P \right\}$. Suppose, without loss of generality, that $\nabla \frac{f(z_p)}{f_0(z_p)} \ne 0$ on $S_P$. Then, $S_P$ has an $(n-1)$-dimensional tangent plane $\textsc{TS}_P $ at $z_P$, and a one-dimensional orthogonal complement  $\langle \nabla \frac{f(z_p)}{f_0(z_p)} \rangle \subset \mathbb{R}^n$. To calculate the integral in \eqref{int1}, we first estimate the region $\Omega_P$. Define the half space $\Omega_T^P = \left\{ \textsc{TS}_P + y\nabla \frac{f(z_p)}{f_0(z_p)}: y \geq 0 \right\}$ and let $\Omega_\delta^P = \left\{ \textsc{TS}_P \pm y\nabla \frac{f(z_p)}{f_0(z_p)}:  0 < y < \delta \right\}$. 
	For $\delta = 2h\lambda$, we have $ \Omega_{T}^{P} \setminus \Omega_{\delta}^{P} \subset \Omega_{P :\lVert u \rVert \leq \lambda} \subset \Omega_{T}^{P} \cup \Omega_\delta^P$. Thus, 
	\begin{equation*}
		\left| \int_{\Omega_P :\lVert u \rVert \leq \lambda } \omega(u)f(z_P+hu)du - \int_{\Omega_T^P:\lVert u \rVert \leq \lambda } \omega(u)f(z_P+hu)du \right|<  \int_{\Omega_{\delta}^P} \omega(u)f(z_P+hu)du. 
	\end{equation*}
	
	Let $M$ and $L$ be the maximum values of $f$ and $\omega$ within the ball $\mathcal{B}_{h\lambda}(z_P)$. Then,
	\begin{equation*}
		\int_{\Omega_{\delta}^P} \omega(u)f(z_P+hu)du <ML \int_{\mathcal{B}_{h\lambda}(z_P)} du \longrightarrow 0 \text{ as }  h \  \longrightarrow 0.
	\end{equation*}
	Therefore, we can choose $\delta_1$ such that $\forall h < \delta_1:\ \int_{\Omega_{\delta}^P} \omega(u)f(z_P+hu)du  < \dfrac{\epsilon}{4}$. 
	On the other hand, using a Taylor expansion, 
	\begin{align*}
		\int_{\Omega_T^P:\lVert u \rVert \leq \lambda } & \omega(u)f(z_P+hu)du = \int_{\Omega_T^P:\lVert u \rVert \leq \lambda } \omega(u)[f(z_P) +\smallO(hu)]du \\
		&= f(z_P)  \int_{\Omega_T^P} \omega(u)du +\smallO(h) -\int_{\Omega_T^P:\lVert u \rVert > \lambda } \omega(u)f(z_P) du\\
		&= \frac{1}{2}f(z_P) - \dfrac{\epsilon}{4}+ \smallO(h).
	\end{align*}
	In the last equation, we can choose  $\delta_2$ such that $\forall h < \delta_2:\ \smallO(h)< \dfrac{\epsilon}{4}$. Hence for $\forall h <\min\{\delta_1,\delta_2 \}$ we have $\left| \int_{\Omega_P} \omega(u)f(z_P+hu)du -\frac{1}{2}f(z_p) \right| < \epsilon$, which completes the proof.
\end{proof}

We now proceed to complete the proof of the theorem. Given that $\frac{1}{2} \frac{f(z_Q)}{f_0(z_Q)} < \frac{1}{2} \frac{f(z_P)}{f_0(z_p)}$, let $ l =\frac{1}{2} \frac{f(z_Q)}{f_0(z_Q)} - \frac{1}{2} \frac{f(z_P)}{f_0(z_p)}$. Then, by Lemma~\ref{lemma4thm5},
\begin{align*}
	\mathbb{P} & \left( \dfrac{\tilde{f}_{z_Q}(z_Q)}{f_0(z_Q)} -  \dfrac{\tilde{f}_{z_P}(z_P)}{f_0(z_p)} > 0 \right)  \leq \\
	& \mathbb{P} \left( \left\{ \dfrac{\tilde{f}_{z_Q}(z_Q)}{f_0(z_Q)} > \dfrac{1}{2} \dfrac{f_{z_Q}(z_Q)}{f_0(z_Q)}+l/2 \right\} \vee 
		\left\{ \dfrac{\tilde{f}_{z_P}(z_P)}{f_0(z_p)} < \dfrac{1}{2} \dfrac{f_{z_P}(z_P)}{f_0(z_p)}-l/2  \right\} \right) \leq \\
	& \mathbb{P} \left(  \dfrac{\tilde{f}_{z_Q}(z_Q)}{f_0(z_Q)} > \dfrac{1}{2} \dfrac{f_{z_Q}(z_Q)}{f_0(z_Q)}+l/2 \right) + 
		\mathbb{P}\left( \dfrac{\tilde{f}_{z_P}(z_P)}{f_0(z_p)} < \dfrac{1}{2} \dfrac{f_{z_P}(z_P)}{f_0(z_p)}-l/2 \right). 
\end{align*}
We can thus write 
\begin{align}\label{eq17}
	\mathbb{P}\Bigg( \dfrac{\tilde{f}_{z_Q}(z_Q)}{f_0(z_Q)} > & \dfrac{1}{2} \dfrac{f(z_Q)}{f_0(z_Q)}+l/2 \mid Q \text{ is chosen}\Bigg) = \nonumber \\
	&\mathbb{P}\left(\frac{1}{n} \sum_{i =1}^{n} h^{-1}\omega((x-z_i)/h)\mathbb{I}_{ \{ z_i \in \Omega_P\}}> \dfrac{1}{2} f(z_Q)+l_1/2 \right),
\end{align}
where $l_1 = lf_0(z_Q)$. Now consider a bandwidth $h$ that satisfies $h= \smallO(1)$; for instance suppose we use the optimal bandwidth $h \sim n^{-\frac{d}{d+4}}$ for $d$ dimensional kernel density estimation. As discussed in the proof of Lemma~\ref{lemma4thm5},  for every $\epsilon \approx h $ we can find a $\lambda$ such that for the truncated kernel $\bar{\omega}(u) = \omega(u) \mathbb{I}_{\{\|u\| < \lambda\}}$ in Lemma~\ref{lemma4thm5}, 
\begin{equation*}
\mathbb{E}[\bar{\omega}{(x-z_i/h)}\mathbb{I}_{ \{ z_i \in \Omega_P\}}] = \dfrac{1}{2} f(z_Q) +\smallO(h).
\end{equation*}
Now, let $X_i = \bar{\omega}{(x-z_i/h)}\mathbb{I}_{ \{ z_i \in \Omega_P\}}$. Then, by the Hoeffding's inequality \cite{BV11} we can rewrite Equation~\eqref{eq17} as

\begin{align*}
\mathbb{P}\left(\bar{X}-\mathbb{E}[\bar{X}]>l_2/2\right) \leq  \exp\left(-\dfrac{nl_2^2}{2c}\right),
\end{align*}
where $l_2 = l_1 + \smallO(h)$. 
Similarly, we have 
\begin{equation*}
	\mathbb{P}\Bigg( \dfrac{\tilde{f}_{z_P}(z_P)}{f_0(z_P)} <  \dfrac{1}{2} \dfrac{f(z_P)}{f_0(z_P)}-l_2/2 \mid P \text{ is chosen} \Bigg)  \leq  \exp\left(-\dfrac{nl_2^2}{2c}\right).
\end{equation*} 
We thus conclude 
\begin{equation*}
	\mathbb{P}  \left( \dfrac{\tilde{f}_{z_Q}(z_Q)}{f_0(z_Q)} > \dfrac{\tilde{f}_{z_P}(z_P)}{f_0(z_p)} \right)  \leq \frac{1}{2} n(n-1) \exp\left(-\dfrac{nl_2^2}{2c}\right),
\end{equation*}
where $\frac{1}{2}n(n-1)$ is an upper bound for the number of combinations that $P$ and $Q$ can take. Therefore, as $n \to \infty$, the above probability converges to zero and the proof is complete. 
\end{proof}

\section{Numerical Experiments}\label{sec:exp}
\subsection{Simulation Studies}\label{sec:sim}
To evaluate the proposed nested region (NR) procedure, we compare its performance with the SC method of \citet{SC07}, the Fisher's method for combining $p$-values and an oracle procedure based on known $f_1$ and $\pi_1$. We consider two simulation scenarios:  
\begin{enumerate}
\item[] \emph{scenario 1}: for each hypothesis, the $d$-variate test statistics are generated independently;
\item[] \emph{scenario 2}: the $d$-variate test statistics are correlated.
\end{enumerate} 
In both settings, the proportion of null hypotheses is set to $\pi_0 = 0.8$ and FDR is controlled at $q = 0.1$ level. For $n \in \{1000,2000,5000,10000\}$ hypotheses, the mean of non-null hypotheses is set to $\mu = 2d^{-1/2}$. 

In the first simulation scenario, data corresponding to each  hypothesis is generated independently, as $\mathbb{N}(\mu,I_d)$, where $\mu = 0$ for null hypotheses and $\mu = 2 d^{-1/2}$ for non-null hypotheses. 
In other words, in this simulation setting, data corresponding to the $d$ dimensions of each of $n$ hypotheses are independent and the $n$ hypotheses are also independent of each other. 

In the second simulation scenario, data for each hypothesis is generated from $\mathbb{N}(\mu,\Sigma_{d\times d})$, where, as before, $\mu = 0$ for null hypotheses and $\mu = 2d^{-1/2}$ for non-null hypotheses. In each simulation instance, the same covariance matrix $\Sigma_{d\times d}$ is used to generate the $d$-variate test statistics corresponding to the each hypothesis; $\Sigma_{d\times d}$ was generated randomly using the \verb=rcorrmatrix= function in the \verb=clusterGeneration= R-package with parameter \verb=alphad= $=1$. In other words, in this simulation setting, the multivariate data for each hypothesis are dependent, while the data for two different hypotheses are independent of each other. 

To apply NR and SC methods, we use the \verb=pdfCluster= R-package \citep{AMR11} to estimate the density of multivariate test statistics using a kernel of the form 
\begin{equation*}
\hat{f}(y)= \sum_{i=1}^n \left( n \prod_{j=1}^d {h^{opt}_{j}} \right)^{-1} \prod_{j=1}^d K\left(\frac{y_{j} - x_{i,j}}{h^{opt}_{j}}\right).
\end{equation*}
In our simulations, the bandwidth for each dimension $j$ is set to the asymptotically optimal bandwidth suggested in \citep{BA97}. 

Table~\ref{tab:case1} and Figure~\ref{fig:case1} show the estimated FDR and FNR of various methods for simulation scenario 1. The estimated rates are based on averages over over $B = 1000$ simulation replicates. 
Table~\ref{tab:case2} and Figure~\ref{fig:case2} show the same estimates for simulation scenario 2. 

In both scenarios, the SC method fails to control the FDR at the desired level, especially as the dimension $d$ increases. 
In scenario 1, the Fisher's method correctly controls the FDR. This is, of course, expected as the underlying independence assumption of the Fisher's method is satisfied in this case. 
In simulation scenario 2, where the $d$-variate test statistics are no longer independent, the Fisher's method fails to control the FDR at the desired level. In fact, the estimated FDR of the Fisher's method is even worst than SC in lower dimensions.

\begin{table}[H]
\caption{The number of rejected hypotheses at various FDR cutoffs using the Fisher's method, the SC method \citep{SC07} and the proposed NR (nested region) method to the TCGA data. The number of rejected hypotheses using }\label{tab:realdata}
\centering
\begin{tabular}{|c|c|c|c|c|c|}
\hline
\, & \multicolumn{3}{|c|}{FDR Adjustment Methods}  & BH--mRNA & BH--DNA 
\\ \cline{2-4}
$q$ & NR  & Fisher & SC & expression & methylation \\ \hline 
0.05      & 52  & 42 & 98 & 16 & 28\\
\hline
0.1    & 129       & 57 & 167 & 32 & 60 \\
\hline
0.15       & 195     & 86 & 224 & 39 & 101 \\
\hline
\end{tabular}
\end{table}

Unlike the SC and the Fisher's methods, the proposed NR method correctly controls the FDR in both simulation settings and for all values of $n$ and $d$. The FNR of the NR method is slightly higher than the oracle procedure and the gap increases with $d$. This is due to the inefficiency of kernel density estimation in higher dimensions. However, the FNR of the NR method improves as the  number of hypotheses $n$ increases. 

\begin{table}[h]
\centering
\caption{Average FDR and FNR of various multiple comparison adjustment methods in simulation scenario 1; here all $n \in \{1000, 2000, 5000, 10000\}$ (null and non-null) hypotheses are  uncorrelated, $q = 0.1$, $\pi_0 = 0.8$, and the mean of non-null hypothesis is set to $2 d^{-1/2}$.}
  \label{tab:case1}
\begin{adjustbox}{max width=0.85\textwidth}
\begin{tabular}{ccccccccccc}
  \hline
  \, & \, & \multicolumn{2}{c}{Fisher} & \multicolumn{2}{c}{SC} & \multicolumn{2}{c}{NR} & \multicolumn{2}{c}{Oracle} \\ 
  $n$ & $d$ & FDR & FNR & FDR & FNR & FDR & FNR & FDR & FNR \\ 
  \hline
1000 &   2 & 0.09 & 0.17 & 0.14 & 0.13 & 0.08 & 0.15 & 0.10 & 0.13 \\ 
  1000 &   3 & 0.09 & 0.18 & 0.23 & 0.12 & 0.09 & 0.15 & 0.10 & 0.13 \\ 
  1000 &   4 & 0.10 & 0.18 & 0.38 & 0.11 & 0.08 & 0.16 & 0.10 & 0.13 \\ 
  1000 &   5 & 0.10 & 0.18 & 0.53 & 0.10 & 0.08 & 0.16 & 0.10 & 0.13 \\ 
  1000 &   6 & 0.10 & 0.18 & 0.64 & 0.08 & 0.08 & 0.17 & 0.10 & 0.12 \\ 
  1000 &   7 & 0.09 & 0.18 & 0.71 & 0.07 & 0.08 & 0.17 & 0.10 & 0.12 \\ 
  1000 &   8 & 0.10 & 0.18 & 0.76 & 0.05 & 0.08 & 0.18 & 0.10 & 0.12 \\ 
  1000 &   9 & 0.09 & 0.19 & 0.79 & 0.04 & 0.09 & 0.18 & 0.10 & 0.12 \\ 
  1000 &  10 & 0.10 & 0.19 & 0.80 & 0.02 & 0.08 & 0.18 & 0.10 & 0.11 \\ \hline
  2000 &   2 & 0.09 & 0.17 & 0.12 & 0.13 & 0.08 & 0.14 & 0.10 & 0.13 \\ 
  2000 &   3 & 0.09 & 0.18 & 0.20 & 0.12 & 0.08 & 0.15 & 0.10 & 0.13 \\ 
  2000 &   4 & 0.09 & 0.18 & 0.34 & 0.11 & 0.08 & 0.15 & 0.10 & 0.13 \\ 
  2000 &   5 & 0.09 & 0.18 & 0.49 & 0.10 & 0.08 & 0.15 & 0.10 & 0.12 \\ 
  2000 &   6 & 0.10 & 0.18 & 0.61 & 0.08 & 0.08 & 0.15 & 0.10 & 0.12 \\ 
  2000 &   7 & 0.10 & 0.18 & 0.69 & 0.07 & 0.08 & 0.16 & 0.10 & 0.12 \\ 
  2000 &   8 & 0.09 & 0.19 & 0.75 & 0.06 & 0.08 & 0.16 & 0.10 & 0.12 \\ 
  2000 &   9 & 0.10 & 0.19 & 0.78 & 0.04 & 0.08 & 0.16 & 0.10 & 0.12 \\ 
  2000 &  10 & 0.09 & 0.19 & 0.80 & 0.03 & 0.08 & 0.17 & 0.10 & 0.11 \\ \hline
  5000 &   2 & 0.10 & 0.17 & 0.11 & 0.13 & 0.08 & 0.14 & 0.10 & 0.13 \\ 
  5000 &   3 & 0.09 & 0.18 & 0.17 & 0.12 & 0.08 & 0.14 & 0.10 & 0.13 \\ 
  5000 &   4 & 0.09 & 0.18 & 0.29 & 0.12 & 0.08 & 0.14 & 0.10 & 0.13 \\ 
  5000 &   5 & 0.09 & 0.18 & 0.44 & 0.10 & 0.08 & 0.14 & 0.10 & 0.13 \\ 
  5000 &   6 & 0.09 & 0.18 & 0.57 & 0.09 & 0.08 & 0.15 & 0.10 & 0.12 \\ 
  5000 &   7 & 0.09 & 0.18 & 0.67 & 0.07 & 0.08 & 0.15 & 0.10 & 0.12 \\ 
  5000 &   8 & 0.09 & 0.18 & 0.73 & 0.06 & 0.08 & 0.15 & 0.10 & 0.12 \\ 
  5000 &   9 & 0.09 & 0.19 & 0.77 & 0.05 & 0.08 & 0.15 & 0.10 & 0.12 \\ 
  5000 &  10 & 0.09 & 0.19 & 0.79 & 0.03 & 0.08 & 0.16 & 0.10 & 0.11 \\ \hline
  10000 &   2 & 0.10 & 0.17 & 0.11 & 0.13 & 0.08 & 0.14 & 0.10 & 0.13 \\ 
  10000 &   3 & 0.09 & 0.18 & 0.15 & 0.13 & 0.09 & 0.14 & 0.10 & 0.13 \\ 
  10000 &   4 & 0.09 & 0.18 & 0.25 & 0.12 & 0.08 & 0.14 & 0.10 & 0.13 \\ 
  10000 &   5 & 0.09 & 0.18 & 0.40 & 0.11 & 0.08 & 0.14 & 0.10 & 0.12 \\ 
  10000 &   6 & 0.09 & 0.18 & 0.54 & 0.09 & 0.08 & 0.15 & 0.10 & 0.12 \\ 
  10000 &   7 & 0.10 & 0.18 & 0.65 & 0.08 & 0.08 & 0.14 & 0.10 & 0.12 \\ 
  10000 &   8 & 0.09 & 0.19 & 0.71 & 0.06 & 0.08 & 0.15 & 0.10 & 0.12 \\ 
  10000 &   9 & 0.09 & 0.19 & 0.76 & 0.05 & 0.08 & 0.15 & 0.10 & 0.12 \\ 
  10000 &  10 & 0.10 & 0.19 & 0.79 & 0.04 & 0.08 & 0.15 & 0.10 & 0.11 \\ 
\hline
\end{tabular}
 \end{adjustbox}
\end{table}

\begin{figure}[h]
\centering
\includegraphics[height=0.485\textheight]{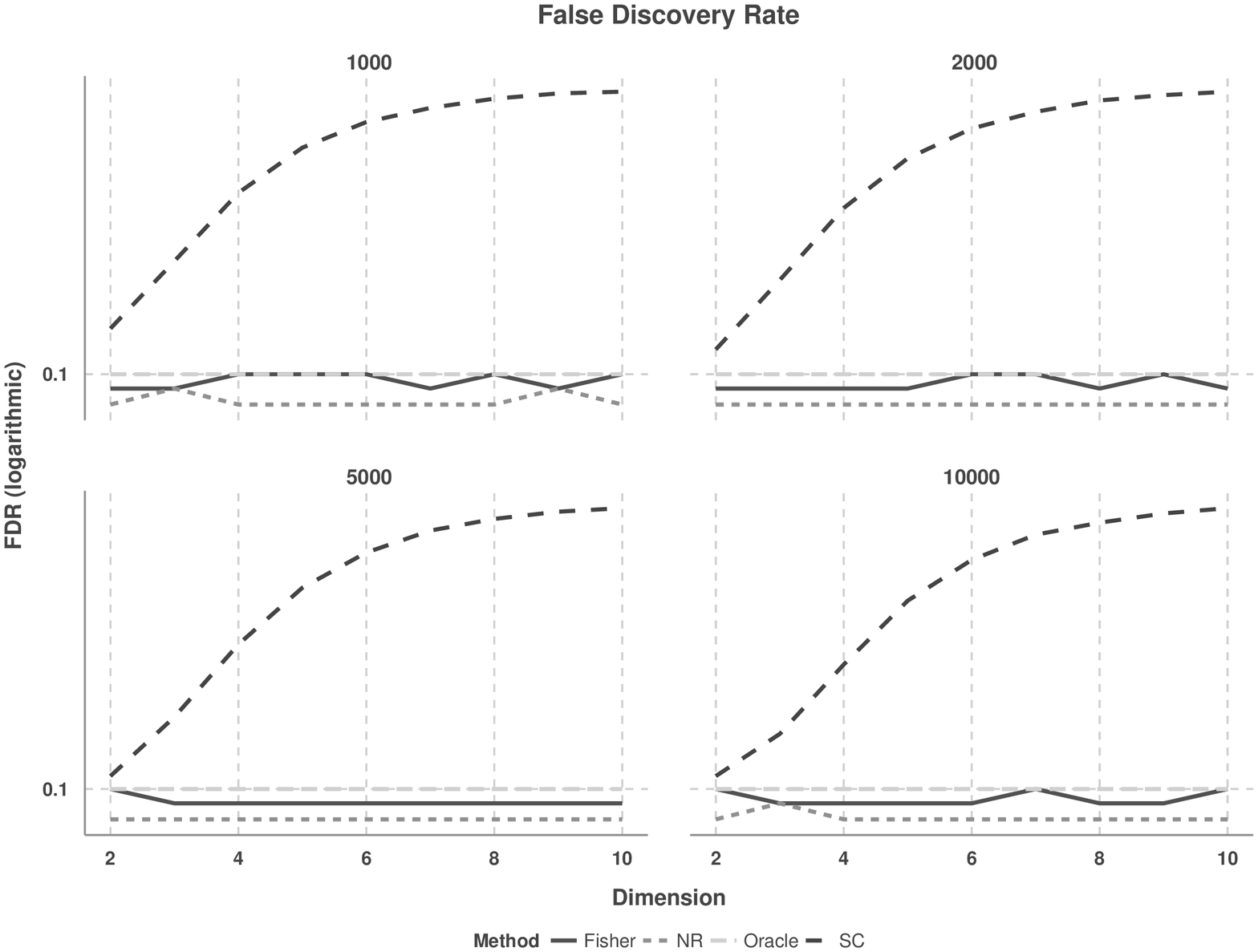} \\
\includegraphics[height=0.485\textheight]{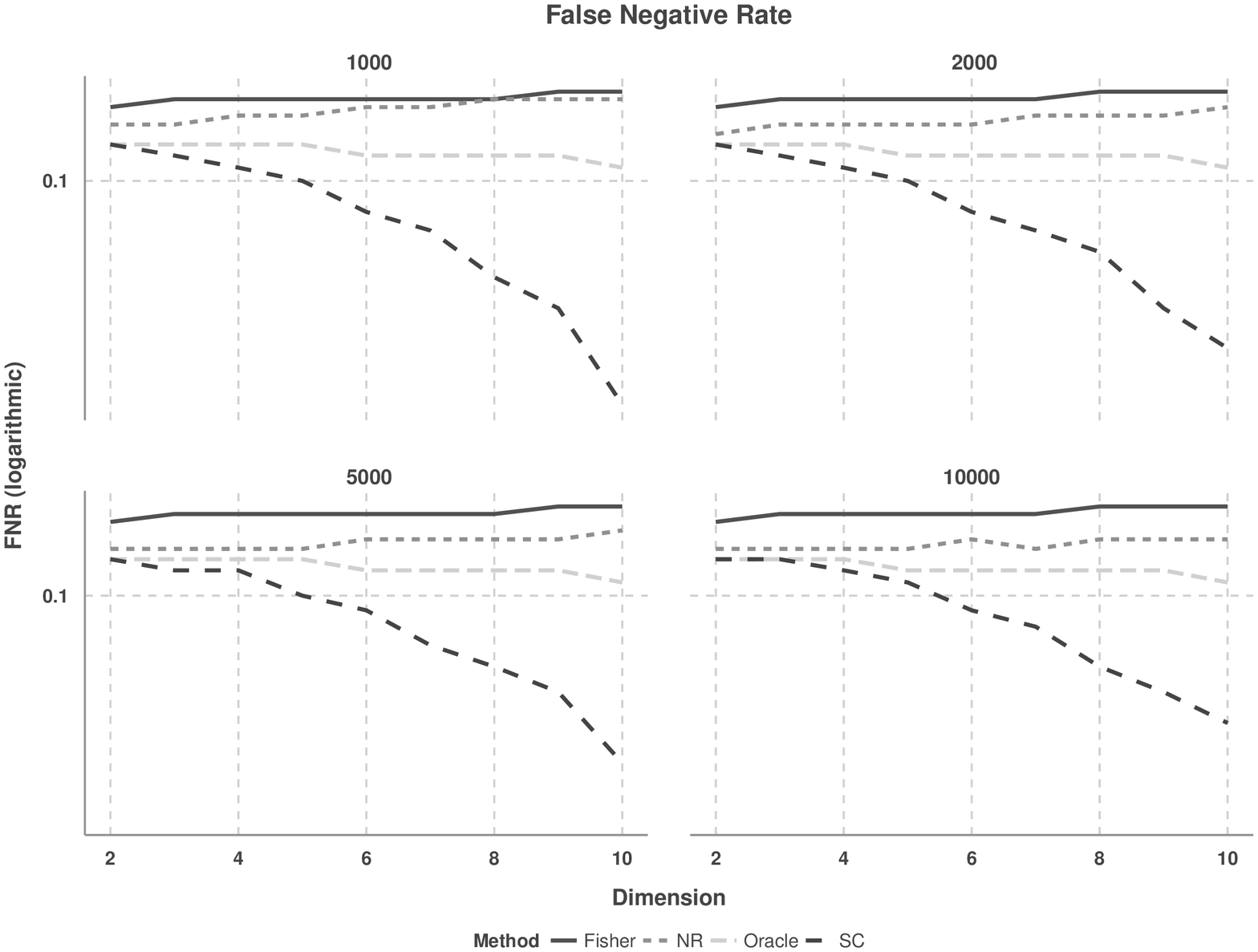}
\caption{Average FDR and FNR of various multiple comparison adjustment methods for the simulation setting of Table~\ref{tab:case1}.}
\label{fig:case1}
\end{figure}

\begin{table}[h]
\centering
\caption{Average FDR and FNR of various multiple comparison adjustment methods in simulation scenario 2; here, $n \in \{1000, 2000, 5000, 10000\}$ $d$-variate test statistics for each hypothesis are generated from a random correlation matrix, $q = 0.1$, $\pi_0 = 0.8$, and the mean of non-null hypothesis is set to $2 d^{-1/2}$.}
\label{tab:case2}
\begin{adjustbox}{max width=0.85\textwidth}
\centering
\begin{tabular}{ccccccccccc}
  \hline
  \, & \, & \multicolumn{2}{c}{Fisher} & \multicolumn{2}{c}{SC} & \multicolumn{2}{c}{NR} & \multicolumn{2}{c}{Oracle} \\ 
  $n$ & $d$ & FDR & FNR & FDR & FNR & FDR & FNR & FDR & FNR \\ 
  \hline
1000 &   2 & 0.33 & 0.13 & 0.13 & 0.13 & 0.08 & 0.16 & 0.10 & 0.13 \\ 
  1000 &   3 & 0.37 & 0.15 & 0.22 & 0.12 & 0.08 & 0.15 & 0.10 & 0.11 \\ 
  1000 &   4 & 0.40 & 0.15 & 0.38 & 0.10 & 0.08 & 0.15 & 0.10 & 0.09 \\ 
  1000 &   5 & 0.42 & 0.16 & 0.53 & 0.09 & 0.08 & 0.15 & 0.10 & 0.07 \\ 
  1000 &   6 & 0.43 & 0.16 & 0.65 & 0.08 & 0.09 & 0.15 & 0.10 & 0.05 \\ 
  1000 &   7 & 0.44 & 0.16 & 0.72 & 0.07 & 0.09 & 0.15 & 0.10 & 0.03 \\ 
  1000 &   8 & 0.45 & 0.17 & 0.77 & 0.07 & 0.08 & 0.15 & 0.10 & 0.02 \\ 
  1000 &   9 & 0.46 & 0.17 & 0.79 & 0.06 & 0.09 & 0.15 & 0.10 & 0.02 \\ 
  1000 &  10 & 0.46 & 0.17 & 0.80 & 0.04 & 0.08 & 0.16 & 0.11 & 0.01 \\ \hline
  2000 &   2 & 0.33 & 0.13 & 0.12 & 0.13 & 0.08 & 0.16 & 0.10 & 0.13 \\ 
  2000 &   3 & 0.37 & 0.15 & 0.19 & 0.12 & 0.08 & 0.15 & 0.10 & 0.11 \\ 
  2000 &   4 & 0.39 & 0.15 & 0.33 & 0.10 & 0.08 & 0.14 & 0.11 & 0.09 \\ 
  2000 &   5 & 0.41 & 0.16 & 0.49 & 0.09 & 0.08 & 0.14 & 0.10 & 0.07 \\ 
  2000 &   6 & 0.43 & 0.16 & 0.62 & 0.08 & 0.08 & 0.13 & 0.10 & 0.05 \\ 
  2000 &   7 & 0.44 & 0.16 & 0.70 & 0.07 & 0.08 & 0.13 & 0.10 & 0.04 \\ 
  2000 &   8 & 0.45 & 0.17 & 0.76 & 0.07 & 0.08 & 0.13 & 0.10 & 0.02 \\ 
  2000 &   9 & 0.46 & 0.17 & 0.79 & 0.06 & 0.08 & 0.13 & 0.10 & 0.02 \\ 
  2000 &  10 & 0.46 & 0.17 & 0.80 & 0.03 & 0.08 & 0.13 & 0.11 & 0.01 \\ \hline 
  5000 &   2 & 0.33 & 0.13 & 0.11 & 0.13 & 0.07 & 0.16 & 0.10 & 0.13 \\ 
  5000 &   3 & 0.37 & 0.15 & 0.16 & 0.12 & 0.07 & 0.15 & 0.10 & 0.11 \\ 
  5000 &   4 & 0.39 & 0.15 & 0.27 & 0.10 & 0.07 & 0.14 & 0.10 & 0.09 \\ 
  5000 &   5 & 0.41 & 0.16 & 0.43 & 0.09 & 0.08 & 0.13 & 0.10 & 0.07 \\ 
  5000 &   6 & 0.43 & 0.16 & 0.57 & 0.08 & 0.07 & 0.13 & 0.10 & 0.05 \\ 
  5000 &   7 & 0.44 & 0.16 & 0.68 & 0.07 & 0.07 & 0.12 & 0.11 & 0.04 \\ 
  5000 &   8 & 0.45 & 0.17 & 0.74 & 0.07 & 0.07 & 0.12 & 0.10 & 0.02 \\ 
  5000 &   9 & 0.46 & 0.17 & 0.78 & 0.06 & 0.07 & 0.13 & 0.10 & 0.01 \\ 
  5000 &  10 & 0.46 & 0.17 & 0.80 & 0.05 & 0.07 & 0.12 & 0.11 & 0.01 \\  \hline
  10000 &   2 & 0.33 & 0.13 & 0.11 & 0.13 & 0.08 & 0.15 & 0.10 & 0.13 \\ 
  10000 &   3 & 0.37 & 0.15 & 0.14 & 0.12 & 0.07 & 0.14 & 0.10 & 0.11 \\ 
  10000 &   4 & 0.40 & 0.15 & 0.24 & 0.10 & 0.07 & 0.14 & 0.10 & 0.09 \\ 
  10000 &   5 & 0.41 & 0.16 & 0.39 & 0.09 & 0.07 & 0.13 & 0.10 & 0.07 \\ 
  10000 &   6 & 0.43 & 0.16 & 0.54 & 0.08 & 0.07 & 0.12 & 0.10 & 0.05 \\ 
  10000 &   7 & 0.44 & 0.16 & 0.65 & 0.07 & 0.07 & 0.12 & 0.10 & 0.04 \\ 
  10000 &   8 & 0.45 & 0.17 & 0.73 & 0.06 & 0.07 & 0.11 & 0.11 & 0.02 \\ 
  10000 &   9 & 0.46 & 0.17 & 0.77 & 0.06 & 0.07 & 0.11 & 0.11 & 0.02 \\ 
  10000 &  10 & 0.46 & 0.17 & 0.79 & 0.05 & 0.06 & 0.11 & 0.11 & 0.01 \\ 
\hline
\end{tabular}
 \end{adjustbox}
\end{table}

\begin{figure}[h]
\centering
\includegraphics[height=0.485\textheight]{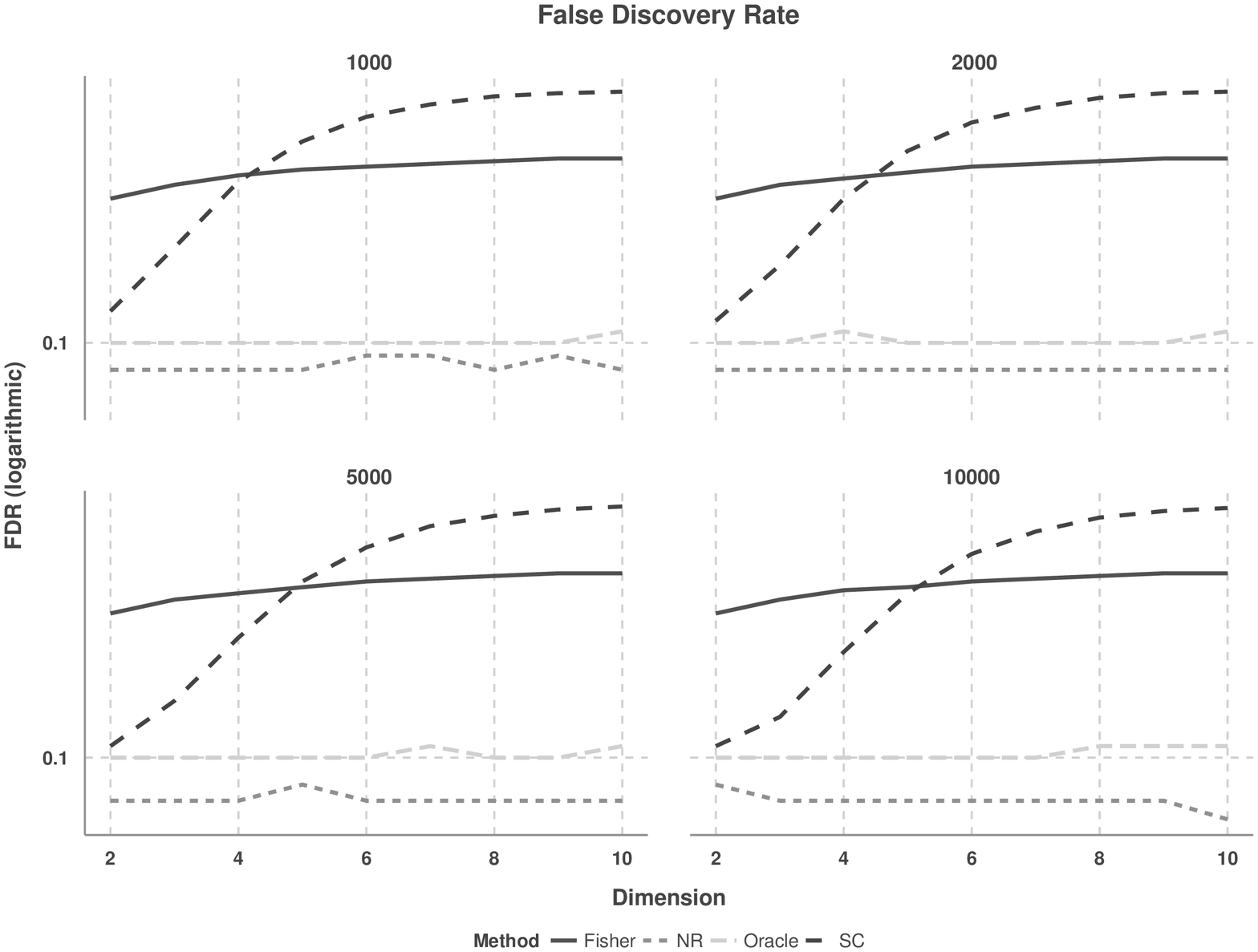} \\
\includegraphics[height=0.485\textheight]{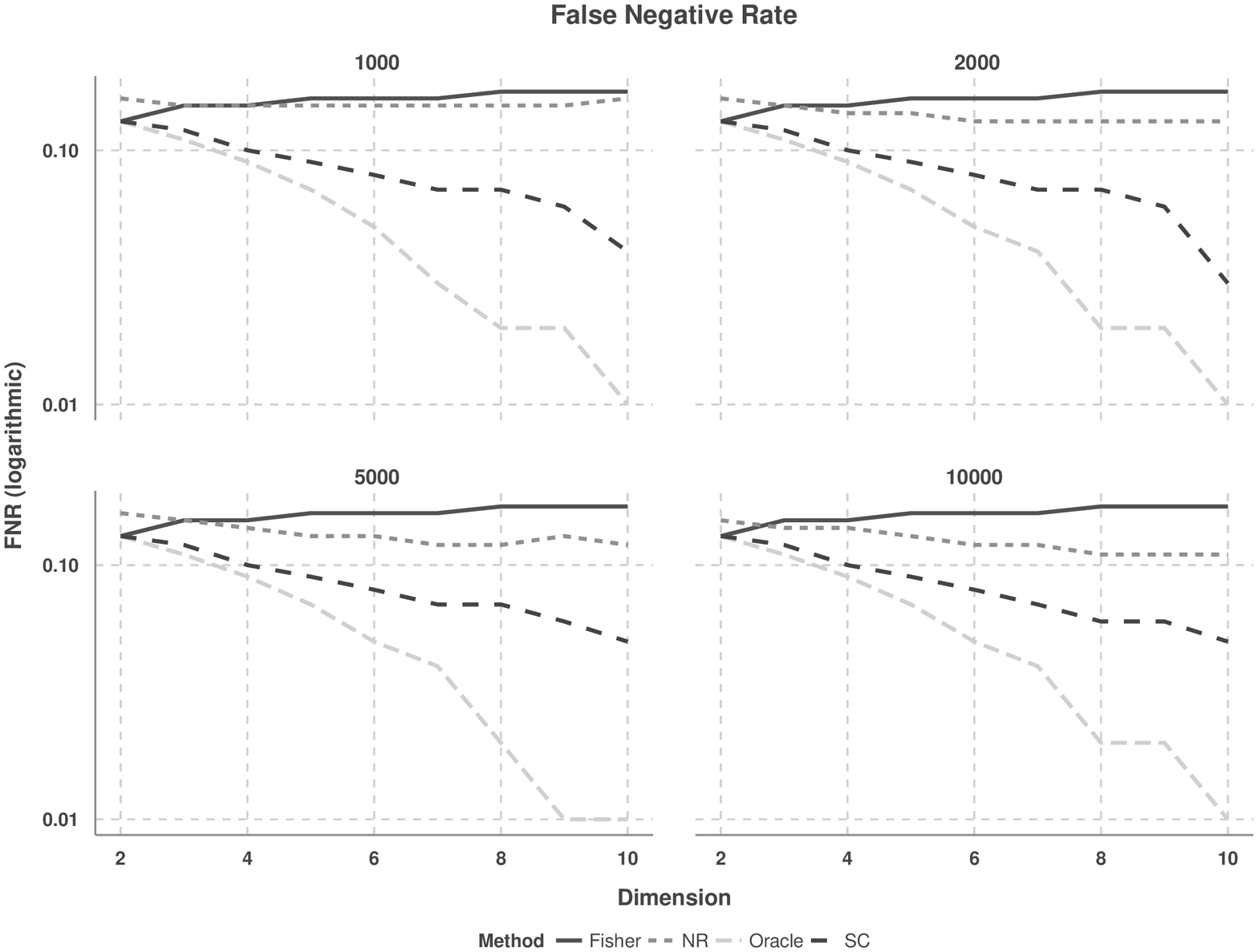}
\caption{Average FDR and FNR of various multiple comparison adjustment methods for the simulation setting of Table~\ref{tab:case2}.}
\label{fig:case2}
\end{figure}

\subsection{Application to TCGA Omics Data}\label{sec:app}
To demonstrate the performance of the proposed nested region (NR) approach to detect non-null multivariate hypotheses, we apply it to \emph{omics} data from The Cancer Genome Atlas (TCGA). 
Specifically, we compare the performance of the NR method with the Fisher's method and the method of \citet{SC07} (SC) to control the false discovery rate when testing for ``differential activity'' of  genes with matched DNA methylation and mRNA expression data from TCGA. 
To this end, ``level 3'' expression and methylation data from 
Colon Adenocarcinoma samples were downloaded from TCGA Data Portal. Matched expression and methylation data for $n=10676$ genes were retained and separate $t$-tests were performed based on each omics data type. DNA methylation and mRNA expression capture two different aspect of epigenetic activity of genes in association with cancer. Combining the evidence from these omics data types can help discover novel associations between genes and cancer status from a data integration perspective. 

Prior to applying our proposed NR method and the competing approaches, we examined the distribution of the test statistics from expression and methylation data and noticed significant deviations from independence. We thus transformed the test statistics to have identity correlation matrix. 

\begin{figure}[t]
  \centering  
  \includegraphics[width=0.75\textwidth, trim=0mm 5mm 5mm 15mm, clip]{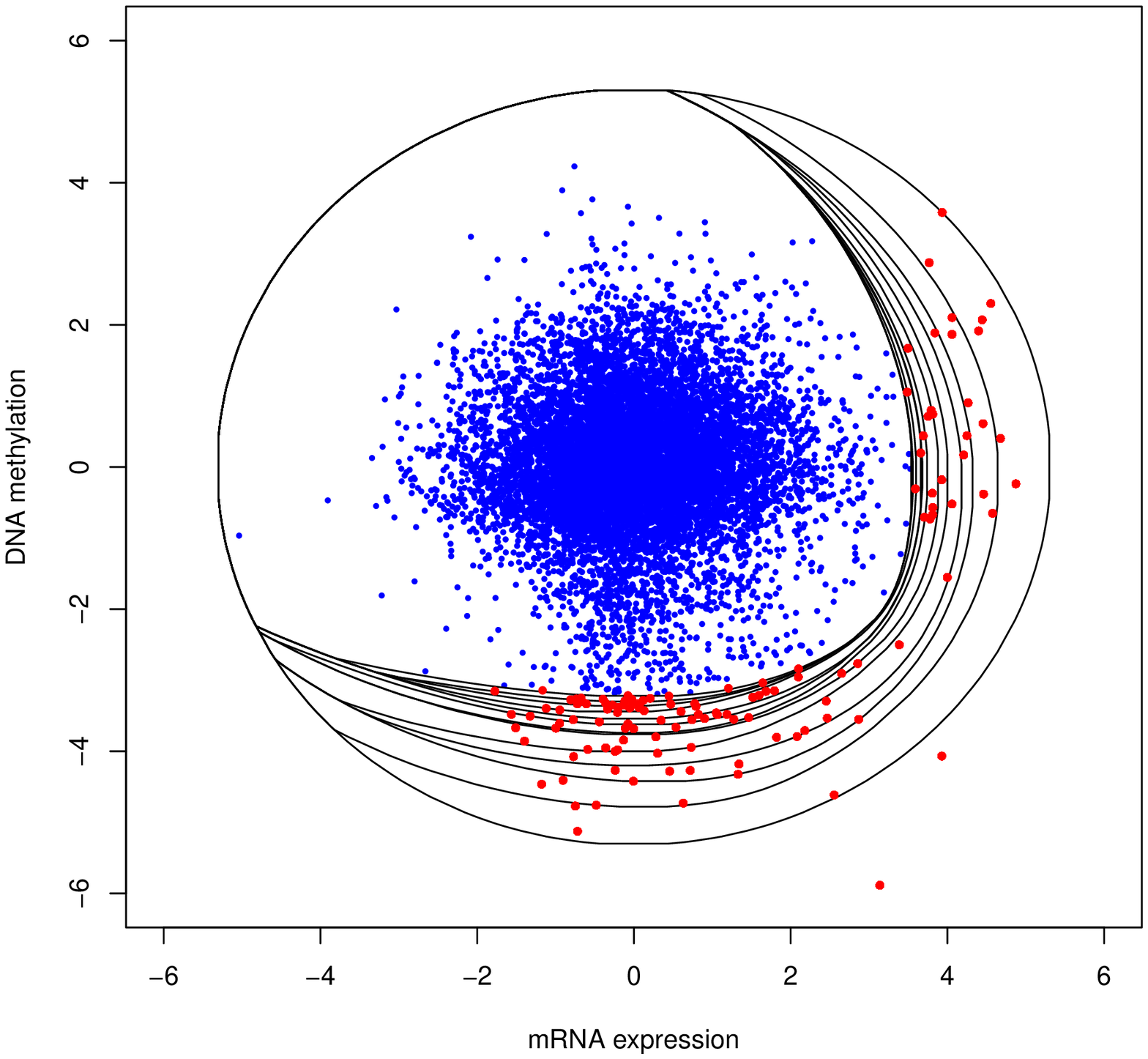}
  \caption{Nested rejection regions for test statistics from the TCGA omics data; small blue circles represent the final null hypotheses, while bigger red points represent non-null hypotheses.}\label{fig:realdata}
\end{figure}

Table~\ref{tab:realdata} shows the number of rejected hypotheses using NR, SC and the Fisher's methods at various FDR cutoffs. 
As pointed out in Section~\ref{sec:sim}, the SC method requires an estimate of the proportion of null hypotheses, $\pi_0$. In Section~\ref{sec:sim}, we used the true value of $\pi_0$ for SC. However, in real data settings, $\pi_0$ is unknown and there are currently no methods for estimating $\pi_0$ in multivariate settings. We thus  estimated $\pi_0$ for SC using the smaller of the two estimated values of $\pi_0$ for the the methylation and expression data. The estimates of $\pi_0$ for each data type were obtained using the  lowest slope line approach of \citet{BH00}; this approach resulted in an estimated $\pi_0$ of 0.9920 for SC.  
Similar to the simulation results in Section~\ref{sec:sim}, SC rejects the largest number of hypotheses, followed closely by NR and then the Fisher's method. However, as we saw in Section~\ref{sec:sim}, there is no guarantee that SC controls the FDR at the desired level. With the transformation applied to the data to obtain an identity correlation matrix, the Fisher's method is expected to control the FDR, as in simulation setting 1 of Section~\ref{sec:sim}; however, our simulation results suggest that the Fisher's method may result in high FNR, which may explain the lower number of rejected hypotheses using the Fisher's method. 

Table~\ref{tab:realdata} also includes the number of hypotheses rejected if only gene expression or DNA methylation evidence is considered (i.e., results from univariate FDR control). The results of the univariate tests suggest that combining the complementary evidence from the two omics sources results in improved power, highlighting the need for methods to control FDR in multivariate hypotheses. 
The nested rejection regions from the proposed NR method for this example are shown in Figure~\ref{fig:realdata}. Here, the form of the rejection regions are close to circles---which is expected given the transformation to identity correlation matrix---and change slightly as the algorithm progresses.

\section{Discussion}\label{sec:disc}
In this paper, a new approach was presented for controlling the false discovery rate (FDR) in multivariate hypothesis testing. 
In todays increasingly data-driven scientific world, multiple evidences are routinely collected for each hypothesis. 
Such multivariate hypotheses offer the opportunity for increased statical power and, hence, new scientific discoveries. Unfortunately, existing approaches are designed for univariate hypotheses and are not guaranteed to control the FDR when testing multivariate hypotheses. On the other hand, combining the multivariate evidence into a univariate summary measure (i.e., a single p-value or univariate test statistic) may result in loss of power. 
The proposed method can thus result in more efficient multivariate hypothesis testing, while controlling the FDR at the desired level. 

The approach proposed in this paper is based on a generalization of the proof of the original \citet{BH95} proposal, which allows for more flexible rejection regions defined based on previously rejected hypotheses. The new proof technique is more broadly applicable and can be used to derive more efficient geometric or algorithmic FDR controlling procedures. 
The nested rejection region (NR) algorithm presented in the paper is an example of such algorithmic approaches, which is shown to be an asymptotically optimal FDR controlling procedure for multivariate hypotheses. 

Similar to the original proposal of \citet{BH95}, our method assumes that the proportion of null hypotheses $\pi_0$ is unknown. It thus controls the FDR conservatively at the level of $\pi_0 q$, instead of the desired FDR level $q$. The estimation of $\pi_0$ for multivariate test statistics is currently an open question and a potentially fruitful research direction. Given such an estimate, more efficient FDR controlling procedures for multivariate hypotheses can be developed. 
Furthermore, our proposal assumes that both null and non-null hypotheses are independent of each other, while allowing for dependence among $d$-variate test statistics corresponding to each hypothesis. 
This assumption is only needed for the development of the optimal algorithm in Section~\ref{sec:approx} and our generalized proof of the BH procedure only requires that the null hypotheses are independent of each other and of non-null hypotheses. 
Extending the theory in Section~\ref{sec:approx} to allow for dependence among multivariate hypotheses and investigating the effect of dependence among hypotheses can be fruitful areas of future research.

\bibliographystyle{plainnat}
\bibliography{FDRrefs}

\end{document}